\documentclass{article}
\usepackage{amsmath}

\usepackage{times}
\usepackage{bm}

\usepackage[plain,noend]{algorithm2e}

\usepackage{amssymb,amsthm}
\usepackage[dvips]{graphicx}
\usepackage{caption}
\usepackage{subcaption}

\usepackage{color}

\makeatletter
\renewcommand{\algocf@captiontext}[2]{#1\algocf@typo. \AlCapFnt{}#2} 
\def\@algocf@capt@plain{top}
\renewcommand{\algocf@makecaption}[2]{%
  \addtolength{\hsize}{\algomargin}%
  \sbox\@tempboxa{\algocf@captiontext{#1}{#2}}%
  \ifdim\wd\@tempboxa >\hsize
    \hskip .5\algomargin%
    \parbox[t]{\hsize}{\algocf@captiontext{#1}{#2}}
  \else%
    \global\@minipagefalse%
    \hbox to\hsize{\box\@tempboxa}
  \fi%
  \addtolength{\hsize}{-\algomargin}%
}
\makeatother

\newtheorem{proposition}{Proposition}[section]

\newtheorem{theorem}{Theorem}[section]
\newtheorem{remark}{Remark}[section]

\newcommand{\pr}{\mbox{pr}}
\newcommand{\var}{\mbox{var}}
\newcommand{\cov}{\mbox{cov}}
\newcommand{\R}{\mathbb{R}}
\newcommand{\N}{\mathbb{N}}

\newcommand{\sumab}[2]{\ensuremath{\sum\limits_{#1}^{#2}}}

\begin{document}

\title{Estimating space-time trend and dependence of heavy rainfall}
\author{A. FERREIRA\\Instituto Superior T\'ecnico da Universidade de Lisboa, 1049-001 Lisbon, Portugal  \\anafh@tecnico.ulisboa.pt\\[.2cm] P. FRIEDERICHS\\Meteorological Institute, University of Bonn, 53121 Bonn, Germany \\pfried@uni-bonn.de\\[.2cm] L. DE HAAN\\Erasmus University Rotterdam, 3000 DR Rotterdam, The Netherlands \\ldehaan@ese.eur.nl\\[.2cm] C. NEVES\\University of Reading, Reading, U.K. \\c.neves@reading.ac.uk\\[.2cm] M. SCHLATHER\\University of Mannheim, 68131 Mannheim, Germany\\schlather@math.uni-mannheim.de}
\maketitle

\begin{abstract}
A new approach for evaluating time-trends in extreme values accounting also for spatial dependence is proposed. Based on exceedances over a space-time threshold, estimators for a trend function and for extreme value parameters are given, leading to a homogenization procedure for then applying stationary extreme value processes. Extremal dependence over space is further evaluated through variogram analysis including anisotropy.
We detect significant inhomogeneities and trends in the extremal behaviour of daily precipitation data over a time period of
84 years and from 68 observational weather stations in North-West Germany. We observe that the trend is not monotonous over time in general.

Asymptotic normality of the estimators under maximum domain of attraction conditions are proven.
\end{abstract}

\vspace{0.5cm} \noindent {\bf Key words}: asymptotic
Extreme precipitation; Extreme value statistics;  Max-stable process; Non-identical distribution; Peaks-Over-Threshold; Trend; Variogram.


\section{Introduction}\label{Introd_sect}
There is some debate on the existence of a significant increase in global  mean temperature (cf. Hawkins et. al. 2017; Hansen, Ruedy, Sato and Lo 2010). Large uncertainties, however, exist with respect to local precipitation in general, and extremes in particular (O'Gorman 2015). 
Our study aims at assessing recent changes in extreme precipitation.
The proposed methods are based on spatial extreme value theory and are developed under maximum domain of attraction conditions. We provide a space-time model that includes trends both in time and space.
Climate change signals in local precipitation extremes are very hard to detect
due to the large variability of precipitation, and its slightly heavy
tail behaviour. Thus many studies investigating single time series of
precipitation hardly find significant signals, while with our spatial analysis we are able to detect some trend behaviour.

Our trend analysis extends earlier results on univariate time series
from de Haan, Klein Tank and Neves (2015) and Einmahl, de Haan and Zhou (2016), 
where the latter introduced a skedasis function to
characterize frequency of high exceedances which, extended to a space-time approach will be the basis for
evaluating trends in time. Further, we extend the results in Einmahl, de Haan and Zhou (2016) to any real extreme value index, while they restricted themselves to the strictly positive case.

Our approach is based on a spatial peaks-over-threshold (POT) method with the novelty of taking one common threshold over space and time. It is developed under maximum domain of attraction conditions,
 differently and broader than what is found in many applications in the field. Most commonly we see max-stable
 models or extreme value copulas (i.e.\ the limiting spatial models) applied directly to data, as in Oesting, Schlather and Friederichs (2016), Buishand, de Haan and Zhou (2008), Davison and Gholamrezaee (2012), Davison, Padoan and Ribatet (2012), and Coles and Tawn (1996) for instance; 
For alternative approaches, e.g. considering
 extreme value analysis with covariates see Friederichs (2010).
 Further, motivated by the work from Oesting, Schlather and Friederichs (2016) for wind data, dependence after homogenization of observations including time-trend removal, is further investigated through a
 parametric power variogram including anisotropy. 

Rainfall observations are available only at discrete points (the stations) and the consideration of a parametric model for describing the dependence structure in space permits to cover the whole space and allows to infer everywhere on high values.

We consider observational weather station data of daily precipitation totals from 1931 until 2014
(84 years in total) from the observing network of the German national meteorological service,
Deutscher Wetterdienst.
As a structural restriction we assume the extreme value index parameter $\gamma\in\R$ constant throughout
space and time. This is a convenient assumption for having manageable theoretical models, also a common restriction
in more applied perspectives, cf. Buishand, de Haan and Zhou (2008) 
and the references therein, and Klein Tank, Zwiers and Zhang (2009). For then we considered the three regions from North-West Germany: Bremen, Niedersachsen and Hamburg with a total of 68 stations, and analyse separately the seasons from November until March and May until September.

\section{Methodology}\label{backtheor_sect}

\subsection{Theoretical Framework}\label{frame_sect}
Consider for each day one has a continuous stochastic process $X=\{X(s)\}_{s\in S}$ representing rainfall over the region $S\subset\R^2$ (a compact state space) and let $X_{i}(s)$ represent daily precipitation totals at locations $s\in S$ and time points $i$ (day). In practice, there are observations at locations $s_j$ ($j=1,\ldots,m$) at each time point $i=1,\ldots,n$ , hence the total number of observed daily precipitation totals is $N=n\times m$.
Let
\begin{equation}\label{margdfx*}
F_{i,s}(x)= \pr \left\{ X_i(s)\leq x\right\},\quad  x^*=\sup\{x:F_{i,s}(x)<1\}\in (0,\infty],
\end{equation}
denote the marginal univariate distribution functions,
supposed continuous with a common right endpoint $x^*$, and
$U_{i,s}=\left\{ 1/(1-F_{i,s}) \right\}^\leftarrow$ (with $\leftarrow$ denoting the
left-continuous inverse function) the associated tail quantile
function. The methods are developed for independent and not
identically distributed observations at discrete points of $X$ in
time that is, the framework is such that the random vectors
$\{X_{1}(s_j)\}_{j=1}^m,\ldots,\{X_{n}(s_j)\}_{j=1}^m$ are
independent  (for details see Section \ref{trenddata_sect} Data and preliminaries) but not necessarily identically distributed. 

Our main structural condition is a spatial POT
condition that includes trend at high levels through the function
$c(s,t)$, with $(s,t)\in S\times[0,1]$ i.e. time 1 to $n$ is
compressed in the interval $[0,1]$. Let $C^+(S)=\{f\in C(S):f\geq
0\}$ be the space of non-negative continuous real functions on S
equipped with the supremum norm and $O$ Borel subsets of $C^+(S)$
satisfying $\inf\{\sup_{s\in S}f(s):f\in O\}>0$. 

Suppose,
\begin{equation}\label{mainscond}
\lim_{u\to\infty}\sup_{n\in\N}\,\sup_{1\leq i\leq n}\,\sup_{s\in S}
\left| u\,\pr \left[ \frac{X_{i}(s)-U_Z \left\{ u\,c(\frac i n,s) \right\}}{a_Z\left\{ u\,c ( \frac i n,s)\right\}}\in O \right]-\pr (
\eta\in O) \right| = 0,
\end{equation}
where $U_Z$ and $a_Z$ are norming constants of a continuous distribution function $F_Z$ with the same common right endpoint from \eqref{margdfx*} and such that,
\begin{equation}\label{maxdomattrZ}
\lim_{u\to\infty} u \left[ 1-F_Z\left\{ U_Z(u)+a_Z(u)x\right\} \right]=(1+\gamma x)^{-1/\gamma},\quad x>0,\;1+\gamma x>0,
\end{equation}
i.e. verifying standard univariate maximum domain of attraction condition for some $\gamma\in\R$.
The function $c$ is a continuous positive function on $S\times [0,1]$, to account specially for trends in time, such that $\sum_{j=1}^m \int_0^1 c(t,s_j)dt=1$, and
 $\eta=\{\eta(s)\}_{s\in S}$ is a stationary Pareto process (Ferreira and de Haan 2014) with marginal tail distribution $(1+\gamma x)^{-1/\gamma}$, $x>0$, $1+\gamma x>0$, $\gamma\in\R$.

Then, relation \eqref{mainscond} implies, 
\begin{multline}
\lim_{u\to\infty}\max_{n\in\N}\,\sup_{1\leq i\leq n}\,\sup_{s\in S}\quad\\ \left| u \left\{ 
  1-F_{i,s} \left(
    U_Z(u)+a_Z \left( u \right) \left[
      \frac{\left\{ c \left(\frac i n,s \right) \right\}^{\gamma}-1}{\gamma} + 
      x\left\{ c\left(\frac i n,s\right)\right\}^{\gamma} \right]
            \right)
\right\}
-(1+\gamma x)^{-1/\gamma}\right|\\=0,\label{1-Frel}
\end{multline}
which is a POT condition verified uniformly in time and space with the appropriate normalization for the right limit i.e. standard generalized Pareto tail. Then \eqref{maxdomattrZ} and \eqref{1-Frel} give that, with $Z$ - $F_Z$ distributed,
\[\frac{Z-U_Z(u)}{a_Z\left(u\right)}, \qquad\frac{X-U_Z(u)}{\left\{c\left(\frac i n,s\right)\right\}^{\gamma}a_Z\left(u\right)}-\frac{1-\left\{c\left(\frac i n,s\right)\right\}^{-\gamma}}{\gamma}
\]
have the same tail distribution, which leads to a way of obtaining a so-called sample of $Z$ (recall a stationary process verifying the maximum domain of attraction condition), from the real observations $\{X_i(s_j)\}_{i,j}$. This gives a homogenization procedure leading to pseudo-observations of $Z$, specified in \eqref{biascorr} later on.

\begin{remark}
    A more appealing condition that still fits to our purposes though theoretically stronger than \eqref{mainscond} is, as $u\to\infty$,
    \begin{multline*}
    \max_{n\in\N}\,\sup_{1\leq i\leq n}\,\sup_{s\in S}\left|\frac{X_i(s)-U_Z(u)}{a_Z\left(u\right)}\right.\\
    \left.-\frac{\left\{c\left(\frac i n,s\right)\right\}^{\gamma}-1}{\gamma}-\frac{Z(s)-U_Z(u)}{\left\{c\left(\frac i n,s\right)\right\}^{-\gamma}a_Z\left(u\right)}\right|
    1_{\{Z(s)>U_Z(u)\}}=o_P(1). 
    \end{multline*}
\end{remark}

To further understand the role of the function $c$ and the process $Z$,
we mention that from the previous conditions it follows that
\begin{equation}\label{XxZUrel}
\lim_{u\to\infty}\max_{n\in\N}\,\sup_{1\leq i\leq n}\,\sup_{s\in S}\left|\frac{U_{i,s}(ux)-U_Z(u)}{a_Z\left(u\right)}-
\frac{\left\{c\left(\frac i n,s\right)x\right\}^{\gamma}-1}{\gamma}\right|=0.
\end{equation}
One sees that the tail quantile function of the original process at
high values, $U_{i,s}(ux)$, can be related to the tail quantile
function of a stationary process verifying maximum domain of
attraction condition at a lower level, $U_Z(u)$, through the trend
function $c$. Since \eqref{XxZUrel} is equivalent to the tail
relation
\begin{equation}\label{trends_cond}
\lim_{x\to x^*} \max_{n\in\N}\,\max_{1\leq i\leq n}\,\sup_{s\in S}\left|\frac{1-F_{i,s}(x)}{1-F_Z(x)}-c\left(\frac i n,s\right)\right|=0,
\end{equation}
it also comes out that the function $c$ is the basis for evaluating
and modelling space-time trends in extremes, since it is seen to
characterize frequency of high exceedances jointly in space and time; for more in the univariate time series context we refer to Einmahl, de Haan and Zhou (2016) and de Haan, Klein Tank and Neves (2015).

\subsection{Estimation of the extreme value parameters and the $c$ function}\label{stsced_sect}
As usual in extreme value statistics, let $k$ be an intermediate sequence, i.e. $k\to\infty$ and $k/n\to 0$ as $n\to\infty$. Let $X_{N-k,N}$ represent the $k$-th upper order statistic from \emph{all} $N=n\times m$ univariate observations $X_{i}(s_j)$. To estimate the shape parameter $\gamma$ consider,
\begin{equation}\label{Moment_def}
\hat\gamma= M_{N}^{(1)} + 1 - \frac{1}{2} \left\{1 - \frac{\left(M_{N}^{(1)}\right)^2}{M_{N}^{(2)}}\right\}^{-1},
\end{equation}
i.e. the same general formula of the well-known moment estimator,
cf. Dekkers, Einmahl and de Haan (1989) but with,
\begin{equation}\label{Mn_def}
M_{N}^{(l)}= \frac{1}{k} \sum_{i=1}^{n} \sum_{j=1}^{m}\left\{\log X_{i}(s_j) - \log X_{N-k,N}\right\}^l 1_{\{X_i(s_j)>X_{N-k,N}\}}\quad (l= 1,2).
\end{equation}
For estimating global location take $\hat U_Z(N / k)=X_{N-k,N}$ and for
estimating global scale,
\begin{equation}\label{scale_estim_def}
\hat a_Z\left(\frac N k\right)= X_{N-k,N} \frac{M_N^{(1)}}{2}\left\{1-\frac{\left(M_{N}^{(1)}\right)^2}{M_{N}^{(2)}}\right\}^{-1}.
\end{equation}
We use similar estimators as if we had independent and identically distributed
observations, but adapted to the main novel characteristic of taking
a unified threshold $X_{N-k,N}$ throughout space and time.

Motivated by Einmahl, de Haan and Zhou (2015) 
for the trend function consider a kernel type estimator,
\begin{equation}\label{chat}
\hat c(t,s_j)=\frac 1{kh}\sum_{i=1}^{n}1_{\{X_i(s_j)>X_{N-k,N}\}}G\left( \frac{t-i/n}{h}\right)\qquad (j=1,\ldots,m;\ t\in [0,1]),
\end{equation}
with $G$ a continuous and symmetric kernel function on $[-1,1]$ such that $\int_{-1}^1 G(s)\,ds=1$, $G(s)=0$ for $|s|>1$; $h=h_n>0$ is the bandwidth satisfying $h\to 0$, $kh\to\infty$, as $n\to\infty$. Similarly as before, the procedure uses a unique threshold throughout space and time.

For estimating $c$ and extremal dependence the intermediate quantity
will be used $C_j(t)=\int_0^t c(u,s_j)\,du, t\in [0,1]$, estimated by,
\begin{equation*}
\hat C_j(t)=\frac 1 k\sum_{i=1}^{nt}1_{\{X_i(s_j)>X_{N-k,N}\}}\qquad (j=1,\ldots,m;\ t\in [0,1]),
\end{equation*}
with $X_{N-k,N}$, $k$ and $N=n\times m$ as before, and $\sum_{j=1}^m \hat C_j(1)=1$.

The proofs of asymptotic normality of the estimators under maximum domain of attraction conditions are postponed to Appendix. 

\subsection{Statistical tests for trend and homogeneity}\label{tests_sect}
The quantities
\begin{equation*}
C_j(1)\quad \text{ and }\quad\frac{c(t,s_j)}{C_j(1)}\qquad (j=1,\ldots,m),
\end{equation*}
give, respectively, time aggregation of high exceedances for location $s_j$ and,
relative space evolution or marginal frequency of high exceedances in time over locations. We use these to test for homogeneity of high exceedances in space and in time, respectively,
\begin{equation}\label{H0J1}
H_{0,j}^{(J1)}: C_j(1)=\frac 1m\qquad (j=1,\ldots,m),
\end{equation}
and
\begin{equation}\label{H0J2}
H_{0,j}^{(J2)}: \frac{C_j(t)}{C_j(1)}=t\qquad (j=1,\ldots,m).
\end{equation}
Asymptotic distributional properties of the corresponding test statistics can be obtained under second order conditions from the results in the Appendix. 

\begin{theorem}\label{asympnormCS_teo}
    Under second order conditions and a convenient growth of the intermediate sequence $k$ ($k\to\infty$ and $n/k\to0$), as
    $n\to\infty$ under $H_{0,j}^{(J1)}$ $(j=1,\ldots,m)$,
        \[
        \surd k\left\{\hat C_j(1)-\frac 1m\right\} 
        {\mathop {\longrightarrow}^{d} }
        W_j\left(\frac 1m\right)-\frac 1m\sum_{i=1}^m W_i\left(\frac 1m\right);
        \]
         under $H_{0,j}^{(J2)}$ $(j=1,\ldots,m)$,
        \begin{multline*}
        \sup_{0\leq t\leq 1} \surd k \left|\hat C_j(t)-t\hat C_j(1)\right|{\mathop {\longrightarrow}^{d} }\\
        \sup_{0\leq t\leq 1}\left|W_j\left\{tC_j(1)\right\}-tW_j\left(C_j(1)\right)
        -tC_j(1)\sum_{i=1}^m\left[W_i\left\{tC_i(1)\right\}-W_i\left\{C_i(1)\right\}\right]\right|,
        \end{multline*}
    with $W_j$ standard Wiener processes.
\end{theorem}

Hence $H_{0,j}^{(Ji)}$, $i=1,2$, will be evaluated through the test statistics
\begin{equation*}
T_j^{(J1)}=\surd k\left| \hat C_j(1)-\frac 1 m\right|,\qquad
j=1,\ldots,m,
\end{equation*}
and
\begin{equation*}
T_j^{(J2)}=\sup_{0\leq t\leq 1}\surd k\left| \hat C_j(t)-t\hat
C_j(1)\right|,\qquad j=1,\ldots,m.
\end{equation*}
Note that the joint limiting structure in both limits in Theorem
\ref{asympnormCS_teo} is left open due to the generality of the main
conditions. That is, as we do not impose any specific joint
structure the joint limiting dependence results unspecified (an
interesting issue beyond the scope of this work and to be
investigated in the future). For applications we propose some approximations as explained in the Data Analysis Section \ref{trenddata_sect} bellow.  

\subsection{Homogenization}\label{homo_sect}
From the tail distribution relations discussed in Section
\ref{frame_sect} we propose the following procedure for having, from
the real observations $\{X_i(s_j)\}_{i,j}$, pseudo-observations of
$\{Z_i(s_j)\}_{i,j}$,
\begin{equation}\label{biascorr}
\hat Z_i(s_j)= \left\{ \hat c\left(\frac i n,s_j\right) \right\}^{-\hat\gamma} X_i(s_j) - \hat a_Z \left( \frac N k \right) 
\frac{ 1- \left\{ \hat c \left( \frac i n,s_j \right) \right\}^{-\hat\gamma} }{\hat\gamma} 
\left\{ 1-\frac{\hat\gamma\hat U_Z \left(\frac N k \right) }{ \hat a_Z \left(\frac N k\right) } \right\},
\end{equation}
$i=1,\ldots,n$, $j=1,\ldots,m$, with $\hat\gamma$,  $\hat a_Z(
N /k)$, $\hat U_Z( N/ k)$ and $\hat c$ the estimators introduced
in Section \ref{stsced_sect}. Note that only the highest
pseudo-observations are considered since the procedure is justified
according to a spatial POT approach with threshold $X_{N-k,N}$.

\subsection{Extremal dependence}\label{dependence_sect}
After homogenizing the data, the modelling concentrates on the extremal spatial dependence further explained from a limiting stationary Generalized Pareto (GP) process (Ferreira and de Haan 2014), 
which we identify with the same dependence structure of the Brown-Resnick process (Brown and Resnick 1977; Kabluchko, Schlather and de Haan 2009) 
with a parametric (semi-)variogram $v(h)$ given by 
\[v(h)= 1/2\ \var \left\{\eta(s+h)-\eta(s)\right\}= 1 /2\ E\left[ \left\{\eta(s+h)-\eta(s)\right\}^2 \right],\]
for separation or lag $h\in\R^2$. The variogram intends to characterize variation in space by measuring evolution of dissimilarities in $\eta(s+h)-\eta(s)$ with lag $h$. 
In order to account for geometric anisotropy we use the power variogram model as in Oesting, Schlather and Friederichs (2014),
$v_{b1,b2,\theta,\alpha}(h)=\Vert A(b_1,b_2,\theta)\,h\Vert^\alpha\ ( h\in\R^2)$,
with $b_1,b_2>0$, $\theta\in(-\pi/2,\pi/2]$ and $\alpha\in(0,2]$,
and the matrix $A$ for geometric anisotropy,
\[
A(b_1,b_2,\theta)=\left(
\begin{array}{cc}
b_1\cos\theta &b_1\sin\theta \\
-b_2\sin\theta &b_2\cos\theta
\end{array}\right).
\]

Variogram estimation is first based in non-parametric estimation of the well-known tail dependence coefficient related to the $L$-dependence function. In the case of Brown-Resnick models the bivariate marginals are known (de Haan and Pereira, 2006; Kabluchko, Schlather and de Haan 2009),
\begin{multline}
L_{s_i,s_j}(x,y)=\lim_{t\to\infty} t\, \pr \left\{ Z(s_i)>U_Z(tx) \vee Z(s_j)>U_Z(ty)\right\}\\
=\frac 1 x\Phi\left[
   \frac{ \{ v_\vartheta(s_i-s_j) \}^{ 1/2 } }{ 2 }
   +\frac{\log\left(\frac y x\right)}{\{v_\vartheta(s_i-s_j)\}^{ 1/2 } }
   \right]
+\frac 1 y\Phi \left[
   \frac{ \left\{ v_\vartheta(s_i-s_j) \right\}^{ 1/2 } }{ 2 }
   +\frac{ \log \left( \frac x y \right) }{ \{v_\vartheta(s_i-s_j)\}^{ 1/2 }  }
   \right].\label{Lij}
\end{multline}
Then,
\begin{equation}\label{L11variogr}
L_{s_i,s_j}(1,1)=2\Phi\left[\frac{\{v_\vartheta(s_i-s_j)\}^{ 1/2 } }{2}\right],
\end{equation}
and an estimator for the variogram is 
\begin{equation}\label{hatgamma}
\hat v(s_i-s_j)=4\left[ \Phi^\leftarrow\left\{\frac{\hat L_{s_i,s_j}(1,1)}{2}\right\}\right]^2
\end{equation}
with $\hat L_{s_i,s_j}(1,1)$ an estimator for the tail dependence function.

Asymptotic normality is well known for the non-parametric estimator,
\[
\hat L_{s_i,s_j}(1,1)=\frac 1{k'} \sum_{l=1}^{n}1_{\{\hat Z_l(s_i)>\hat Z_{n-k',n}(s_i) \text{ or } \hat Z_l(s_j)>\hat Z_{n-k',n}(s_j)\}}
\]
where $\hat Z_{n-k',n}(s)$ denotes the $(n-k')$th order statistic from $\{\hat Z_i(s)\}_{i=1}^{n}$ and $k'$ is an intermediate sequence ($k'\to\infty$, $k'/n\to 0$ as $n\to\infty$).
Under suitable conditions (cf. Einmahl, Krajina and Segers 2012)
\[
\surd{k}\left\{\hat L_{s_i,s_j}(1,1)-L_{s_i,s_j}(1,1)\right\}\to B_{s_i,s_j}(1,1),
\]
with $B_{s_i,s_j}(1,1)$ zero-mean Gaussian distributed, which should lead to, by the delta method and \eqref{L11variogr},
\begin{multline*}
\surd{k}\left\{\hat v(h)-v_\vartheta(h)\right\} 
{\mathop {\longrightarrow}^{d} } 
B_h(1,1)\frac{2v_\vartheta(h)^{1 /2}}{\phi\left\{\frac{v_\vartheta(h)^{ 1 /2}}2\right\}},
\end{multline*}
with $h\in\R^2$ and $\phi$ the standard normal density.

Finally, the variogram parameter estimates are obtained numerically,
through
\begin{equation*}
\min_{b1,b2,\theta,\alpha}\sum_{1\leq i,j\leq m}\left\{\hat v(s_i-s_j)-v_{b1,b2,\theta,\alpha}(s_i-s_j)\right\}^2
\end{equation*}
with $\hat v(s_i-s_j)$ from \eqref{hatgamma} and
\begin{eqnarray*}
\nonumber   v_{b1,b2,\theta,\alpha}(h)&=&  \sumab{i=1}{2}\sumab{k=1}{2} d_{ik}\, h_ih_k\\
\nonumber                        &=& (b_1h_1\cos \theta + b_1h_2 \sin \theta)^2 + (b_2h_2 \cos \theta -b_2 h_1 \sin \theta)^2\\
&=&  b_1^2h_2^2 + b_2^2h_1^2 + (b_1^2-b_2^2)\left\{ \frac{h_1^2-h_2^2}{2}\, \cos(2\theta) + h_1h_2 \sin(2\theta)\right\},
\end{eqnarray*}
i.e.\ using $(||Ah||)^2= h^{T}Dh$ with $D$ a symmetric matrix with entries $d_{11}= b_1^2\cos^2\theta +b_2^2\sin^2\theta$, $d_{12}= (b_1^2-b_2^2)/2\, \sin (2\theta)$ and  $d_{22}= b_1^2\sin^2 \theta + b_2^2 \cos^2 \theta$.

\section{Data Analysis}\label{applic_sect}

\subsection{Data set and preliminaries}\label{data_sect}
The considered data amounts to daily precipitation totals from the
observing network of the German national meteorological service,
Deutscher Wetterdienst, Offenbach (available at \linebreak 
ftp://ftp-cdc.dwd.de/pub/CDC/observations$_{-}$germany/climate/). Mostly all observational weather
stations from the three regions in North-West Germany, Bremen,
Niedersachsen and Hamburg, with available observations from (at
least) 1931 until (at least) 2014 were selected, after preliminary
data analysis. In total we ended up with $m=68$ stations over $n=84$
years. Two seasons were considered separately, a cold season from
November until March and a warm season from May until September.
For brevity we shall mostly concentrate on the results for the cold season.

In preparation of the data the following issues were taken into account, for coherence with the theory, namely in applying the proposed spatial POT methods on the basis of a time independent sample:
\begin{enumerate}
\item {\it Independent observations.} In practice it is considered that serial data as daily precipitation totals is approximately independent after one or two days; Caires (2009), de Haan, Klein Tank and Neves (2015).
In order to avoid losing extremal information we constructed an
approximately independent sample from the initial time series as
follows. Order the sample maximum of all observed processes and pick
up the process with the maximum value, say $\Vert
X\Vert_{n,n}=\max_{1\leq i\leq n}\left\{\max_{1\leq j\leq m}
X_i(s_j)\right\}$. Then observe the second maximum \linebreak $\Vert
X\Vert_{n-1,n}$: discard it if it is within two days lag from the
previous one otherwise keep it. Observe the next one $\Vert
X\Vert_{n-2,n}$: again discard it if it is within a lag of two days
from any of the retained processes, otherwise retain it. Continue
this procedure until reaching the desired number of higher
processes for the data analysis.
It is worth mentioning, though not being our present target, these
estimates of the extreme value index are more stable after this
procedure, hence allowing for lower values of $k$ in threshold
selection, which should be particularly useful for small sample
sizes.

\item {\it POT method.} The described procedure in selecting the highest observations is coherent with the spatial POT approach (Ferreira and de Haan, 2014),
namely with $Z=\{Z(s)\}_{s\in S}$ stationary being such that
\begin{multline*}
\lim_{t\to\infty} \pr \left[
  \left\{ 1+\gamma\frac{Z-b_Z(t)}{a_Z(t)}\right\}_+^{1/\gamma}\in A  \mid \sup_{s\in S}\left\{ 1+\gamma\frac{Z-b_Z(t)}{a_Z(t)} \right\}_+^{1/\gamma}>1
  \right]
= \pr (W\in A),
\end{multline*}
with $W$ a simple Pareto process (cf. their Theorem 3.2), and thus in agreement with the
previously described criteria for large values in terms of
maxima of each observed process. 
Of course in previous point 1.\ we
are not at the $Z$'s yet, but in general extreme values of $X$ are
passed to the de-trended sample as seen in Figure \ref{boxplotsW}, cf.\ (a) and (b) where indeed the highest values remain in the new homogenized sample.
\end{enumerate}

\begin{figure}
 \includegraphics[width=14cm]{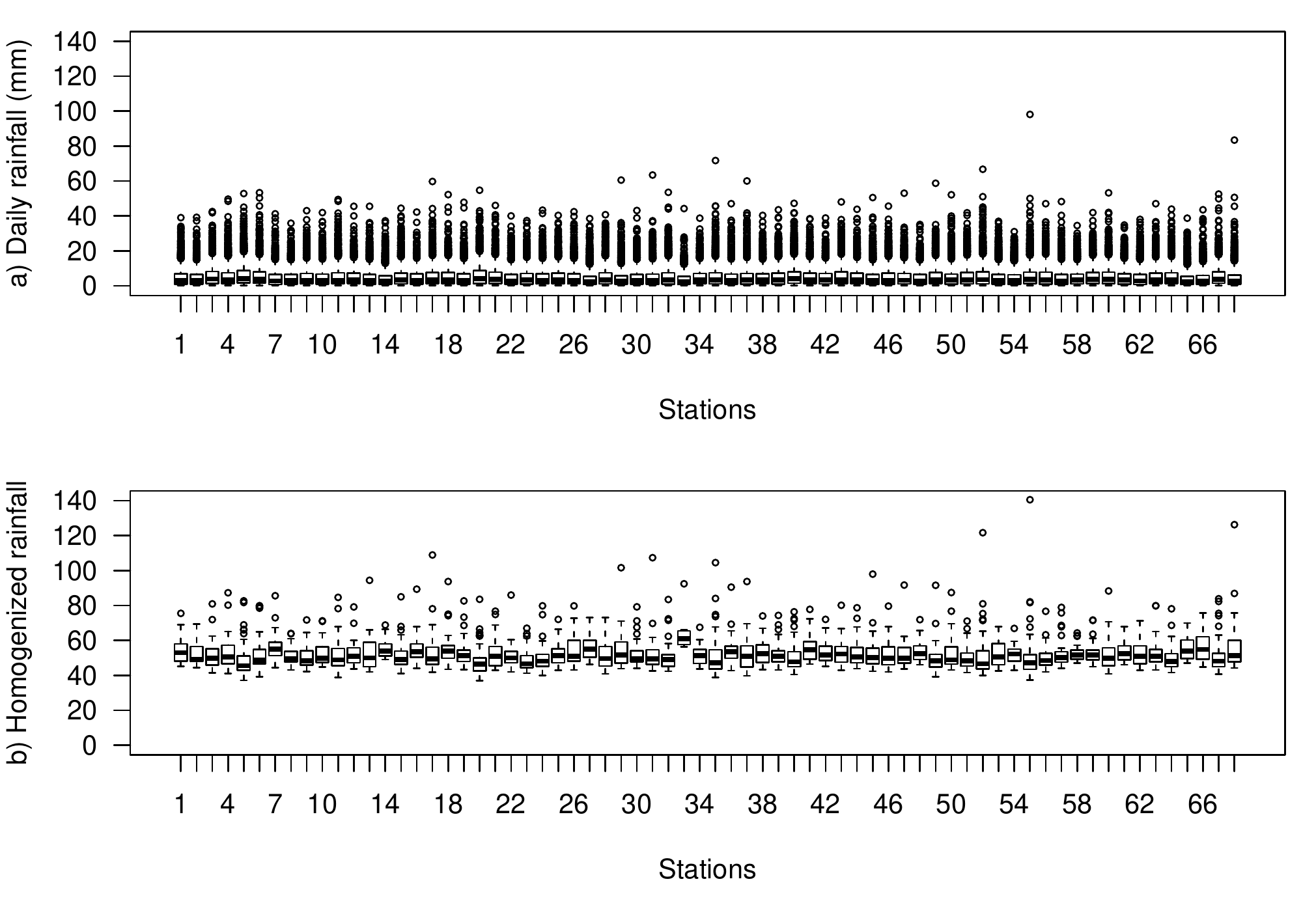}
\caption{Box-plot of daily precipitation in mm (excluding zeros) during the cold season at 68 stations for (a) observed data, and (b) the corresponding homogenized data.} 
\label{boxplotsW}
\end{figure}

\subsection{Trend and homogenization analysis}\label{trenddata_sect}

For estimating the extreme value parameters $\gamma$ and $a_Z$ as a
function of the number of top order statistics $k$, according to
\eqref{Moment_def}-\eqref{Mn_def} and \eqref{scale_estim_def} respectively, after some graphical analysis also taking into account the estimation $C$, we opted for $k=3\,000$ for the cold and $k=4\,000$ for the warm season giving the
parameter estimates shown in Table \ref{gabestimates.table}. When comparing cold to warm season, as should be expected, one gets higher estimates for the warm season but, on the other hand the cold season seemed more interesting when concerning trends. As already mentioned we shall mostly concentrate on the cold season.
\begin{table}
    \caption{Top sample size $k$, as well as shift, shape and scale estimates}
    \label{gabestimates.table}
    \centering
    \begin{tabular}{rrrr|rrrr}
         \multicolumn{4}{c} {\emph{Cold season}}& \multicolumn{4}{c} {\emph{Warm season}}\\
         $k$ & $X_{N-k,N}$ & $\hat\gamma$ & $\hat a_Z$ & $k$ & $X_{N-k,N}$ & $\hat\gamma$ & $\hat a_Z$  \\
        \hline
        3000 & 21$\cdot$3 & 0$\cdot$07 & 4$\cdot$98 & 4000 & 27$\cdot$4 & 0$\cdot$13 & 9$\cdot$41
    \end{tabular}
\end{table}

In Fig. \ref{spatialchatkW} are shown the estimated scedasis functions over time. 
We have used the biweight kernel $G(x)=(15/16)(1-x^2)^2$, $x\in [-1,1]$, with a boundary correction. 
There seems to exist a general tendency for the increasing of high values with time in the cold season, though with large fluctuation. 

\begin{figure}
		\includegraphics[width=14cm,height=7cm]{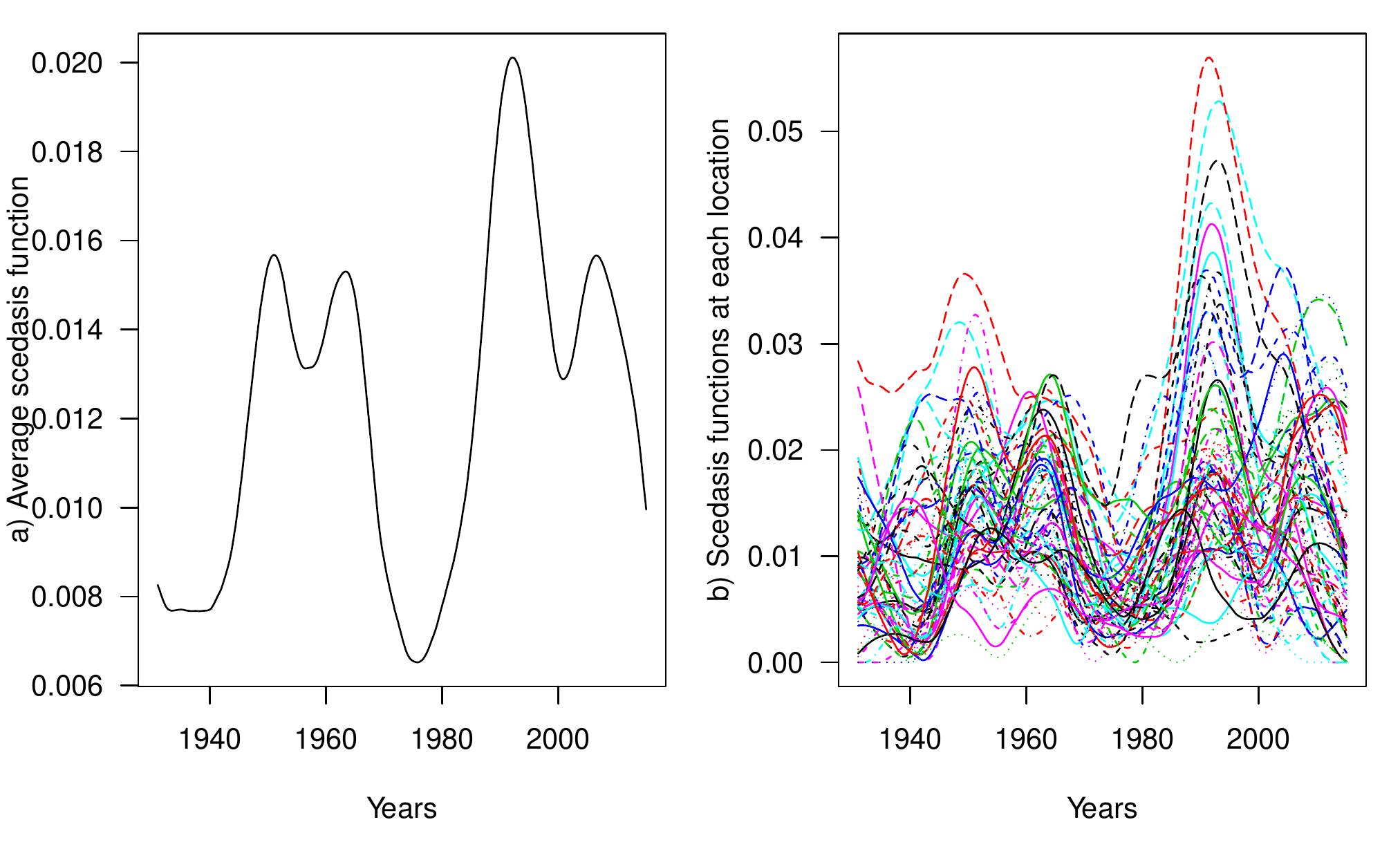}
	\caption{Estimated scedasis functions $c$ for the cold season. a) Average over all estimated scedasis functions at the different locations, and b) estimated scedasis at each location. The parameter $k$ is set to $k=3000$.}
	\label{spatialchatkW}
\end{figure}

While applying the statistical tests $H_{0,j}^{(Ji)}$, $i=1,2$, for the rejection criteria we aimed at a most conservative approach. First regard that the limiting distributions in Theor \ref{asympnormCS_teo} in case of independence among $W_j$ $(j=1,\ldots,m)$ correspond to $N(0,(1-1/m)/m)$ in the first case and to $N(0,t(1-t)C_j(1)\{1-tC_j(1)\})$ for each $t$ and $j$, in the second case. Then for the first test statistic we used as approximate limiting distribution, with
\[(mk)^{-1/2}\frac{\hat C_j(1)-1/m}{2(1-1/m)}\] 
as being approximately normal, following Lehmann and Romano (2005) to cover the maximal possible variance for improvement on the power of the test, combined with Bonferroni correction. Similarly, for obtaining the approximate distribution for second test statistic, we used the approximate distribution for the supremum of Brownian bridge (Hall and Wellner 1980) with maximal variance calculated from the limiting distributions at each $t$ given in Theorem \ref{asympnormCS_teo}, again combined with Bonferroni correction.

An alternative approximation but less conservative approach, would be to take for the limiting distributions normal zero mean with the variances mentioned before under independence.

The results of the statistical tests $H_{0,j}^{(Ji)}$, $i=1,2$ at a
95\% confidence level are shown in Fig. \ref{testsW}.
Stations Bodenfelde-Amelith (station 5) and Bad Iburg (station 40) both in Niedersachsen at altitudes 258 and 517 meters, respectively, identified with red circles in (a), have significant high number of exceedances. Although we see some increasing trends in the scedasis functions, it has large variability so that trends in time are not significantly detected.

\begin{figure}
		\includegraphics[width=.48\linewidth]{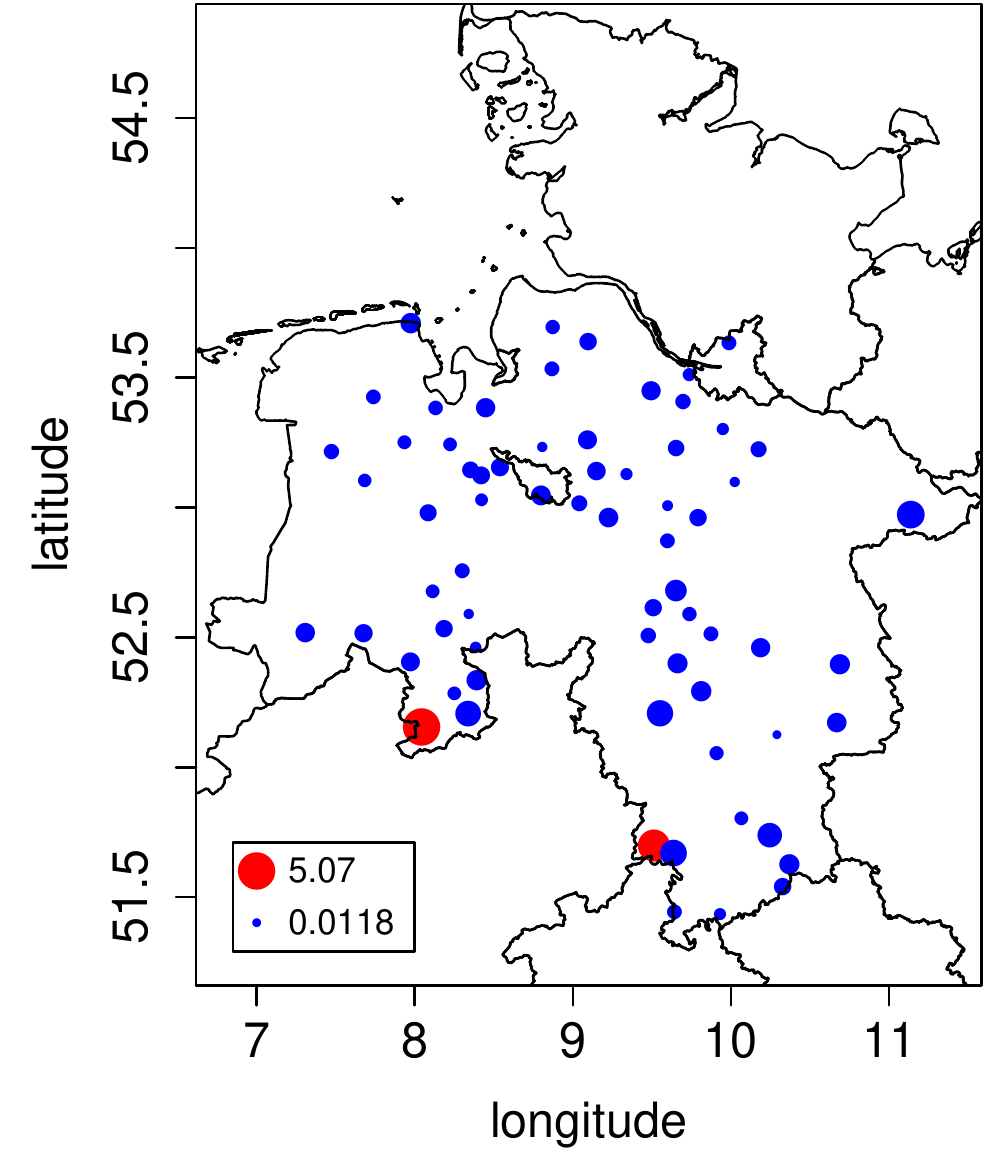}
		\includegraphics[width=.48\linewidth]{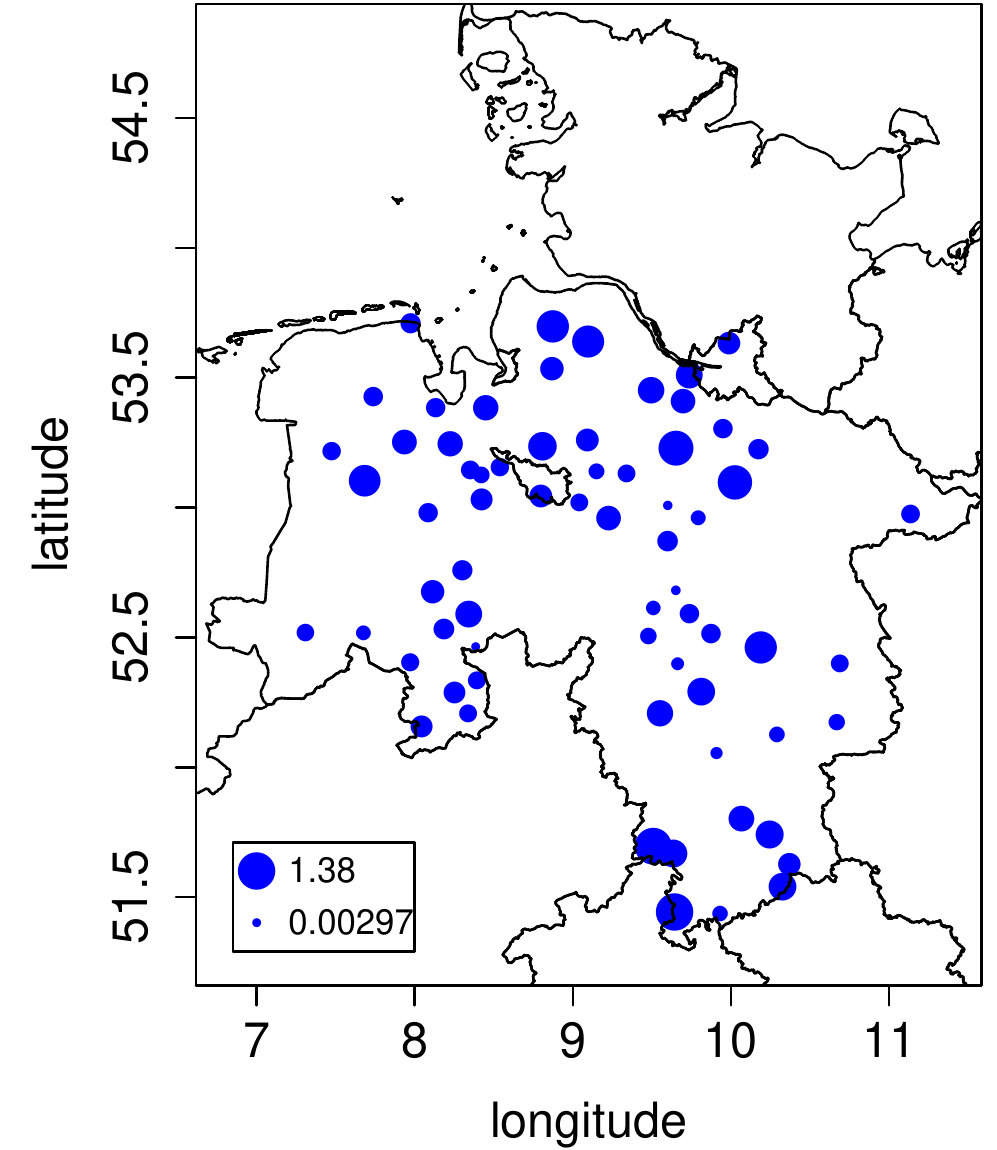}
	\caption{Statistical tests for $H_{0,j}^{(J1)}$ (left panel) and  $H_{0,j}^{(J2)}$ (right panel) for 68 meteorological weather stations. The size of dots is proportional to the value of the respective test statistic. Significant values are indicated in red, non-significant values in blue.}
	\label{testsW}
\end{figure}

The estimates of the function $c$ combined with $\hat\gamma$, $\hat
a_Z$ and $\hat U_Z$ from Table \ref{gabestimates.table} are used for
homogenization \eqref{biascorr}.  The procedure while restricting to
high values performs a zooming on these, hopefully meaning
that information from highest values is more efficiently used.

\subsection{Spatial dependence analysis}\label{dependencedata_sect}
The estimation of the dependence structure of the homogenized data needs the estimation of the tail dependence coefficient $L(1,1)$, as explained in Section \ref{dependence_sect}. There is a threshold choice for $\hat L_{s_i,s_j}(1,1)$, denoted by $k'$, since the de-trended sample $\hat Z_i(s_j)$ is truncated. Only those $(i,j)$ for which $X_i(s_j)>X_{N-k,N}$ are considered or, in other words, for each time series $j$ we have $\hat C_j(1)\times k$ exceedances, with $k$ the total number of order statistics used for estimation, being equal to 3000 for the cold and 4000 for the warm season. Since $L_{s_i,s_j}(1,1)$ is to be estimated for all pairs $(i,j)$, to avoid extra bias due to truncation a possible choice is $k'=\min_{j=1,\ldots,68}\hat C_j(1)\times k$. We slightly deviated from this in trying to be more efficient and considered different $k'$'s by taking $k'_{i,j}=\min\{\hat C_i(1),\hat C_i(1)\}\times k$  (and being asymptotically negligible). This avoids losing sample information and accordingly seems to slightly improve the dependence estimation.

The estimated $L_{s_i,s_j}(1,1)$ against the Euclidean distance between stations $(s_i,s_j)$ as well as the contour plot of the corresponding estimated variogram, are shown in Fig. \ref{paramvariograms}, where the variogram parameter estimates after numerical
minimization are summarized in Table \ref{varestimates_tb}. 
In general spatial tail dependence looks stronger in the cold season than in warm season at smaller distances, although one observes large variability in the estimates. The contour plot indicates stronger dependence through North-South direction.
\begin{figure}
	\centering
		\includegraphics[width=.48\linewidth]{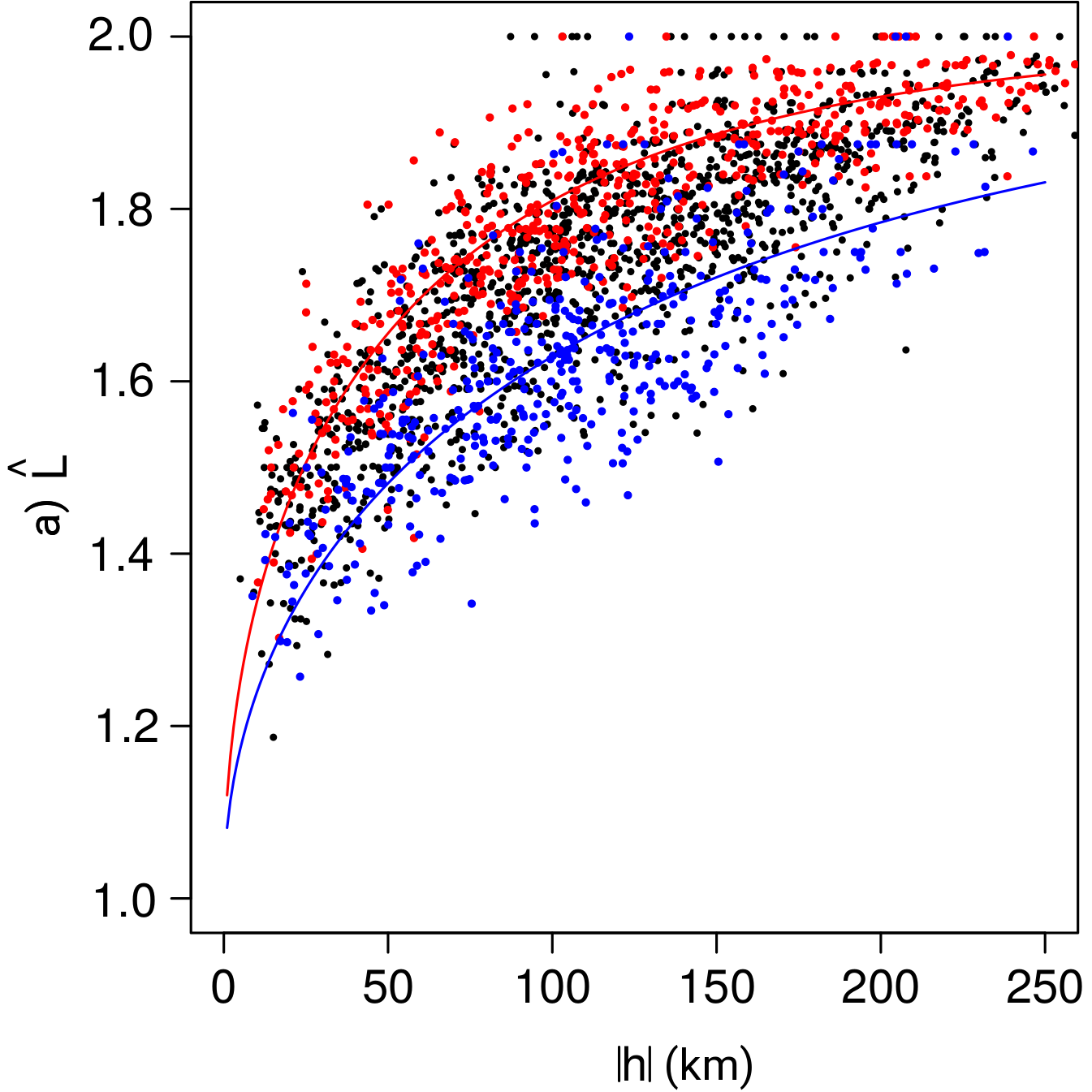}
		\includegraphics[width=.48\linewidth]{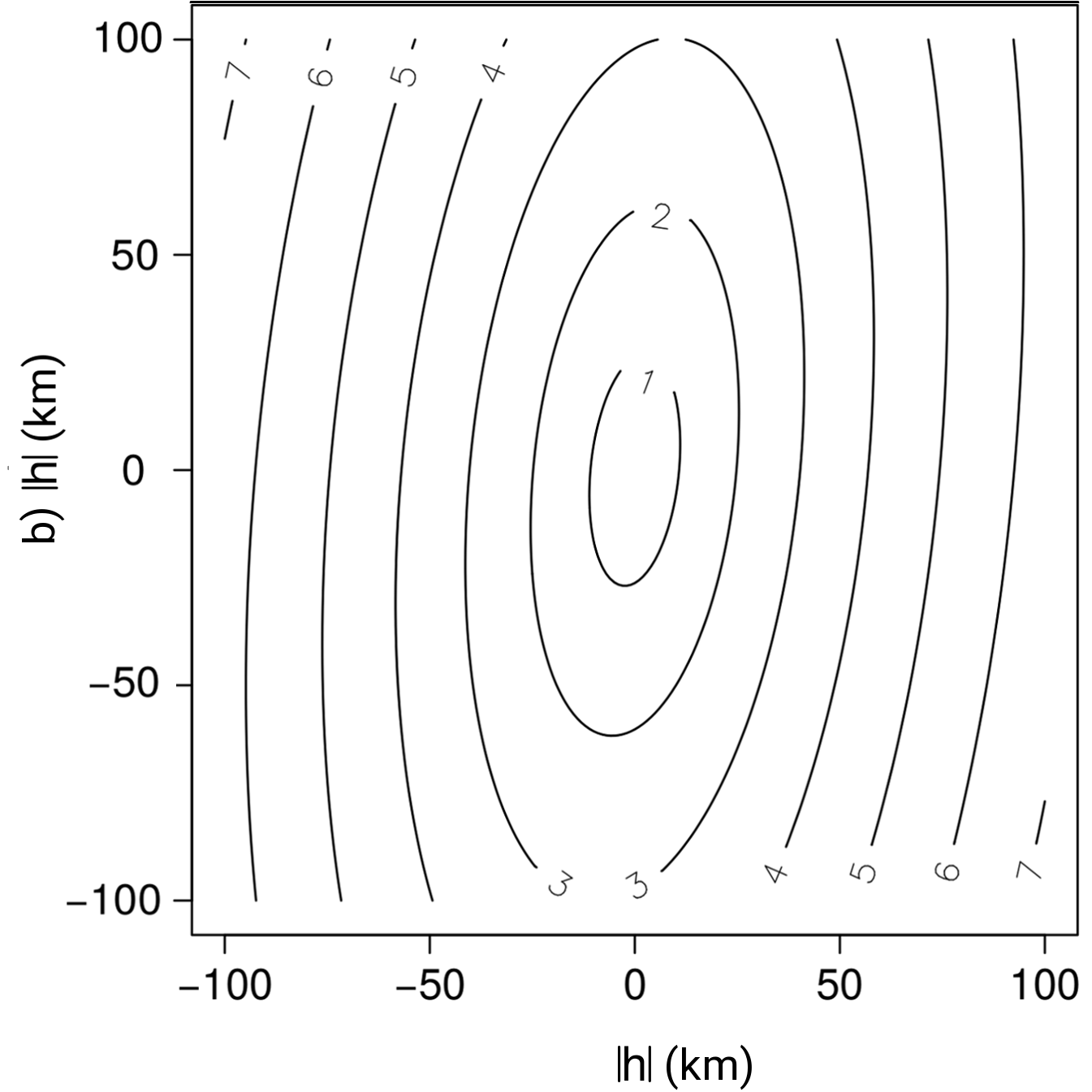}
	\caption{Tail dependence coefficient estimates and respective variogram fit for cold season. a) Tail dependence coefficient with respect to the principal axis of the anisotropy matrix A: red to all distance vector that have an angle with
the first axis of at most $\arccos(0\!\cdot\!95)$, red line the theoretical variogram in the direction
of the first axis; in blue analogously for the second axis; in black all the remaining points. b) Contour level plot of the fitted parametric variogram in the horizontal Euclidian plane.} 
	\label{paramvariograms}
\end{figure}

\begin{table}\caption{Variogram parameter estimates}\label{varestimates_tb}
	\centering
	\begin{tabular}{c|ccc|ccc}	
		& \multicolumn{3}{c} {{Cold season}}& \multicolumn{3}{c} {{Warm season}}\\
	Parameters	& Estimate & Std. Error & p.value  & Estimate & Std. Error & p.value \\
		\hline
		$\hat{b}_1$ & 0$\cdot$26919 & 0$\cdot$01679 & $<2\times10^{-16}$ &  1$\cdot$9252 &  0$\cdot$1155 & $<2\times10^{-16}$\\
		$\hat{b}_2$ & 1$\cdot$13446 & 0$\cdot$08742 & $<2\times10^{-16}$ &  1$\cdot$8290 &  0$\cdot$1177 & $<2\times10^{-16}$\\
		$\hat{\theta}$ & 0$\cdot$09214 & 0$\cdot$01189 & 1$\cdot$47$\times10^{-14}$& 0$\cdot$7854 &  0$\cdot$2337 & 0$\cdot$000789 \\
		$\hat{\alpha}$ & 0$\cdot$85579 & 0$\cdot$02415 & $<2\times10^{-16}$ & 0$\cdot$6844 &  0$\cdot$0124 & $<2\times10^{-16}$ \\
    \end{tabular}
\end{table}

\subsection{Failure probability estimation}\label{failureprobab_sect}
Finally the proposed models are applied in failure (or
exceedance) probability estimation in univariate and
bivariate settings. In the univariate setting (i.e. for a given
location $s$ and time $t$) this probability is defined as,
\[
p_n(t,s)= \pr \left\{X_t(s)>x_n\right\}
\]
for a given high value $x_n$, usually a value that none or few observations have exceed it (asymptotically $x_n$ should approach $x^*$ as $n\to\infty$).
One of the classical estimators on the basis of an independent and identically distributed  sample of random variables, say $Y_1,\ldots,Y_n$, and $k=k(n)$ an intermediate sequence ($k\to\infty$ and $k/n\to 0$, as $n\to\infty$), is
\begin{equation}\label{defhatpn}
\hat p_n=\frac k n \left( 1+\hat\gamma_n\frac{x_n-Y_{n-k,n}}{\hat \sigma_n}\right)^{-1/\hat\gamma_n},
\end{equation}
where $\hat\gamma_n$ and $\hat\sigma_n$ are suitable estimators according to the maximum domain of attraction condition for $F_{Y_1}$.

As a way of extending the independent and identically distributed setting, we want to take into account trend information. From \eqref{trends_cond} one has a relation among exceedance probability of $X_i(s_j)$, trend function $c$ and exceedance probability of $Z$. Combining these we propose the following to estimate failure probabilities over time,
\begin{equation}\label{defhatcpn}
\hat p_n\left(t,s_j\right)=\hat c\left(t,s_j\right)\frac k N \left\{ 1+\hat\gamma\frac{x_n-\hat Z_{n-k,n}}{\hat a_Z(\frac N k)}\right\}^{-1/\hat\gamma}.
\end{equation}
We do not have available all pseudo-observations $\hat Z_i(s_j)$, since the procedure for obtaining these is justified only for high observations. However, this should not pose a problem as our methods are for high values.

In Fig. \ref{failprobst5} are represented several curves for $\hat p_n\left(t,s_j\right)$ with $x_n=40$ and for Bodenfelde-Amelith (station 5) and Uslar (station 55), using the corresponding estimated function $\hat c\left(t,s_j\right)$ ($j=5,55$ respectively represented in Fig. \ref{spatialchatkW}) for different values of $k$. We mention that failure probability estimates obtained as if the samples $\{X_{i}(s_j)\}_i$, $j=5,55$, were independent and identically distributed and using \eqref{defhatpn}, seem to overestimate probabilities and clearly show larger variance and bias specially for the last period of time.

\begin{figure}
	\centering
		\includegraphics[width=.40\linewidth]{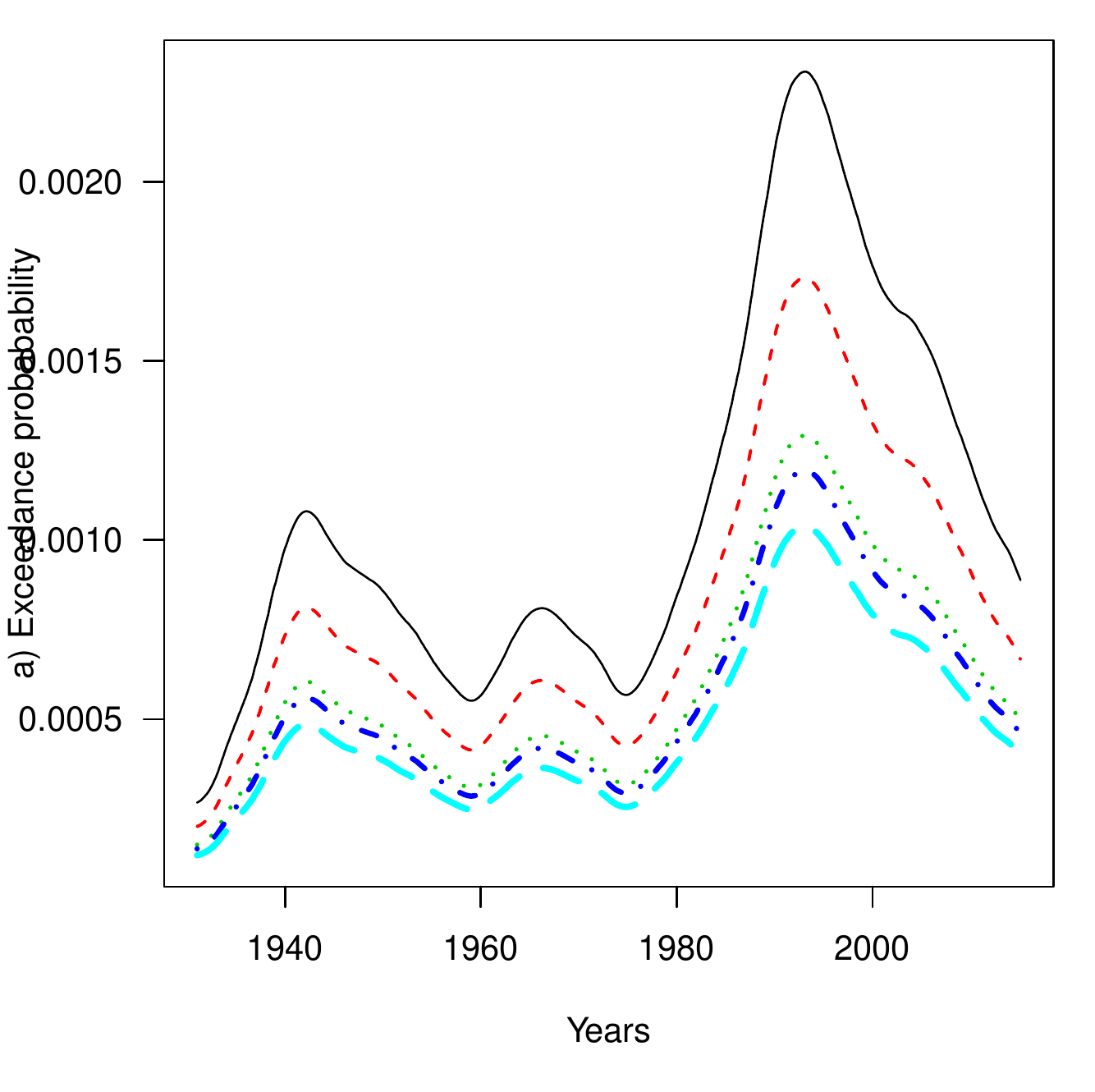}
		\includegraphics[width=.40\linewidth]{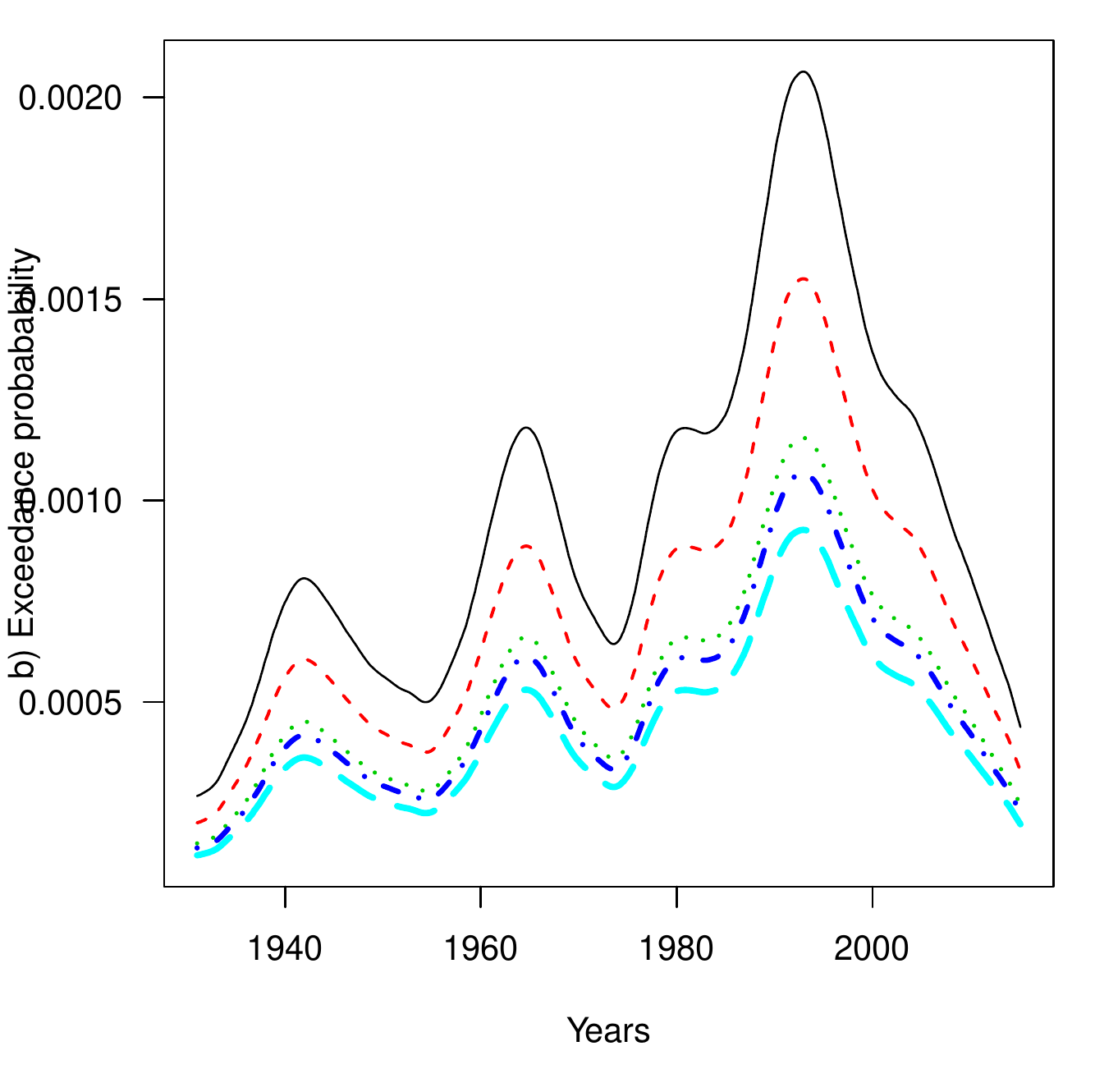}
		\includegraphics[width=.40\linewidth]{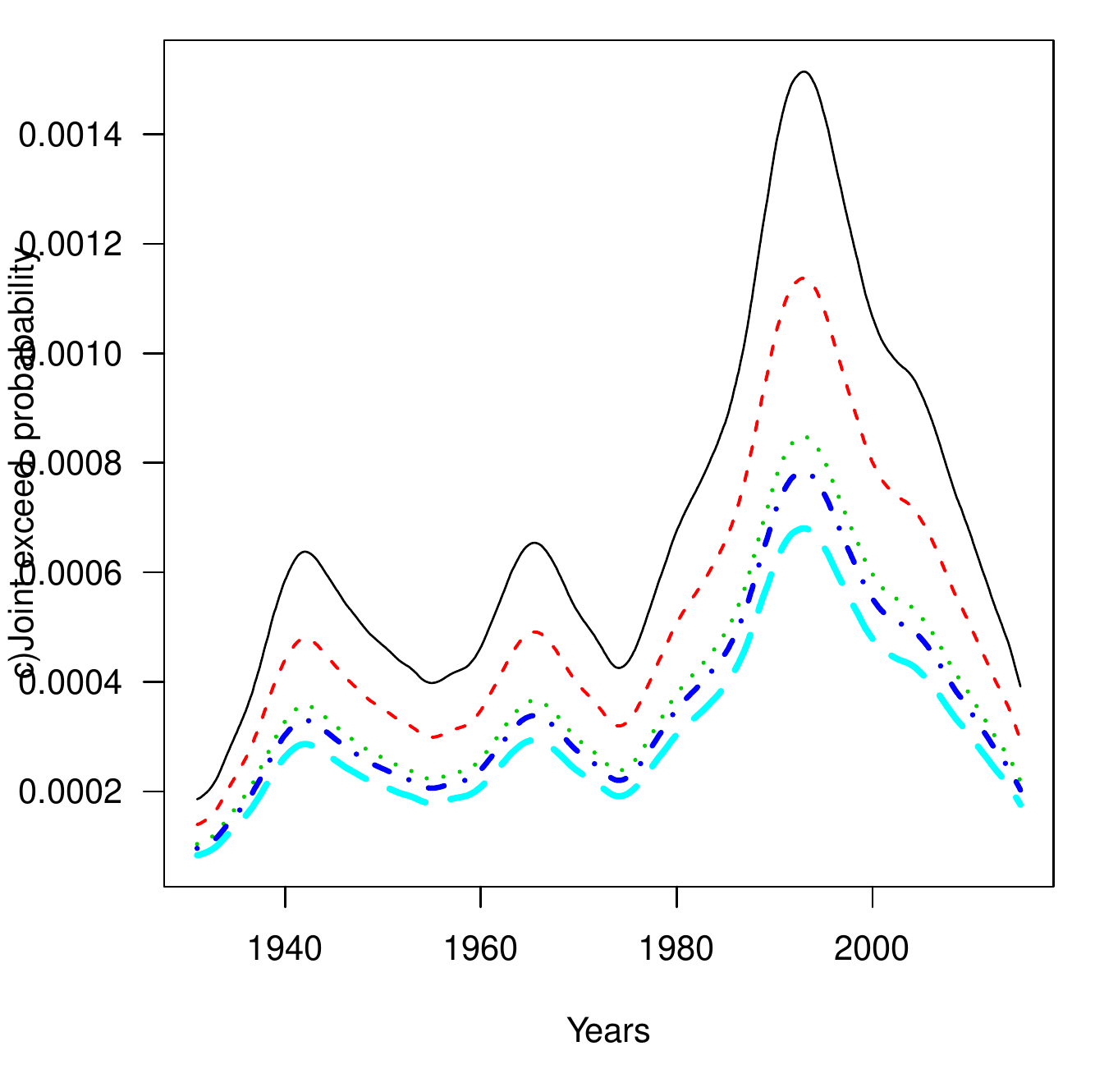}
		\includegraphics[width=.40\linewidth]{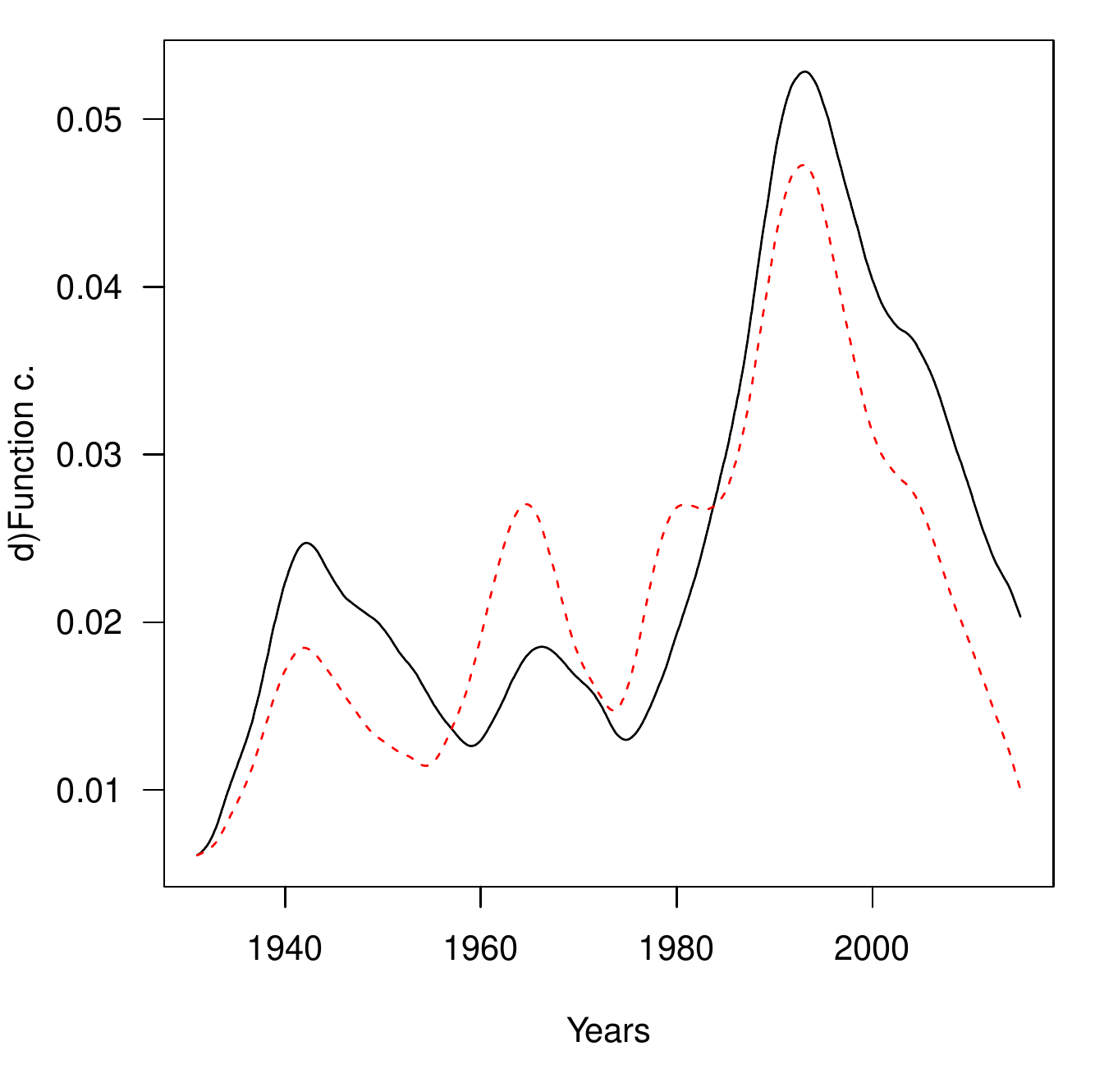}
	
	\caption{Estimated exceedance probabilities over time of a) station 5 and b) station 55, c) joint exceedance probability and d) function $\hat c\left(t,s_j\right)$, $j=5,55$. Colored lines in a)-c) indicate $k$ with $k=600$ (black line),  $k=1000$ (red), $k=1400$ (green),  $k=1800$ (dark blue), and $k=2200$ (light blue), in d) joint exceedance probability for station 5 (black line) and station 55 (red dotted line).}
\label{failprobst5}
\end{figure}

Joint failure probability estimation is also shown in Fig. \ref{failprobst5} for different values of $k$ for Bodenfelde-Amelith and Uslar station. The estimates are obtained by combining \eqref{Lij} including variogram estimates and marginal estimates. That is for estimating,
\[
p_n\left(t,s_5,s_{55}\right)=\pr\left\{X_t(s_5)>40\, \wedge \, X_t(s_{55})>40\right\}
\]
take,
\[
\hat p_n\left(t,s_5,s_{55}\right)=\hat p_n\left(t,s_5\right)+\hat p_n\left(t,s_{55}\right)
-\frac k N \hat L_{s_{5},s_{55}}\left\{\frac N k p_n^{-1}\left(t,s_5\right),\frac N k p_n^{-1}\left(t,s_{55}\right) \right\}.
\]
As expected the results are a combination of the previous marginal estimates, giving a smaller probability of exceedance.

\section*{Appendix}
Let $F_Z$ be a continuous distribution function with right-endpoint $x^*=\sup\{x:F_{Z}(x)<1\}\in (0,\infty]$, and
$U_{Z}=\left\{1/(1-F_{Z})\right\}^\leftarrow$ (with $\leftarrow$ denoting the
left-continuous inverse function) the associated tail quantile
function.

Suppose there exists a positive eventually decreasing function $\alpha$, with $\lim_{t\to\infty}\alpha(t)=0$ and $a_Z$ such that for all $x>0$, $\gamma\in\R$,
\begin{equation}\label{UZ2ndord_cond}
\lim_{t\to\infty}\frac{ \frac{U_Z(tx)-U_Z(t)}{a_Z(t)}-\frac{x^\gamma-1}{\gamma} }{\alpha(t)}=\int_1^x s^{\gamma-1}\int_1^s u^{\rho-1} du\,ds =H_{\gamma,\rho}(x),\quad\rho<0.
\end{equation}

Let $\gamma_+ = \gamma\vee 0$ and $\gamma_-= (-\gamma) \vee 0$ 
for any $\gamma\in \R$.

\begin{proposition}\label{1-FOA_prop}
  If $F_Z$ satisfies \eqref{UZ2ndord_cond}, suitable functions
  $a_0(t)>0$, $b_0(t)\in\R$ and $\alpha_0(t)\to 0$ exist such that
\begin{equation*}
	\sup_{x\geq x_0>-1/\gamma_+}\left(\frac{t\left[1-F_Z\{b_0(t)+xa_0(t)\}\right]}{(1+\gamma x)^{-1/\gamma}}-1\right) \alpha_0(t)^{-1}=O(1)\qquad (t\to\infty).
	\end{equation*}
\end{proposition}

\begin{proof}
 Proposition 3.2 in Drees, de Haan and Li (2005) (cf. Theorem 5.1.1, p.~156, in de Haan and Ferreira 2006), 
         yields that 
	 for  $\varepsilon>0$, $\gamma\in\R$ and $\rho\leq 0$,
         but not $\gamma=0=\rho$, 
	\begin{equation}\label{1-FunifDHL}
	\left( \frac{t\left[1-F_Z\{b_0(t)+xa_0(t)\}\right]}y-1\right)\alpha_0(t)^{-1}=y^\gamma\Psi_{\gamma,\rho}(y^{-1})+y^{-\rho}e^{\varepsilon |y|}o(1),
	\end{equation}
	with $y=(1+\gamma x)^{-1/\gamma}$.
        The $o(1)$-term holds uniformly for $x\geq x_0>-1/\gamma_+$. Further, 
	\begin{equation}\label{Psibar_def} 
	\Psi_{\gamma,\rho} (x)=\left\{\begin{array}{ll}
	\frac{x^{\gamma+\rho}}{\gamma+\rho} \;, &\gamma+\rho\neq 0 \;, \; \rho<0 \\
	\log x \;, &\gamma+\rho = 0 \;, \; \rho<0 \\
	\frac{1}{\gamma} x^\gamma \log x \;, &\rho = 0\neq\gamma, 
	\end{array} \right.
	\end{equation}
	and $a_0$, $b_0$ and $\alpha_0$ are such that
	\begin{equation}\label{abA0}
	\frac{a_0(t)}{a_Z(t)}=1+O\{\alpha(t)\},\quad \frac{b_0(t)-U_Z(t)}{a_0(t)}=O\{\alpha(t)\},\quad \frac{\alpha_0(t)}{\alpha(t)}=O(1).
	\end{equation}
	These functions can be obtained from Drees (1998) and Cheng and Jiang (2001), the ones for obtaining second order regular variation uniform bounds (c.f. also de Haan and Ferreira 2006).
	It remains to verify that the right-hand side in \eqref{1-FunifDHL} is bounded for all $x\geq x_0>-1/\gamma_+$, which follows for $\rho<0$ by straightforward calculations. 
\end{proof}

For establishing the asymptotic distribution of $\hat C_j(t)$ the
following second order condition is also needed, which relates, for
each $s$,
 the marginal distribution functions of t$X(s)$ with that
 of the (not observable) $Z(s)$, jointly
satisfying the maximum domain of attraction condition.

Suppose a positive eventually decreasing function $A$  exists with $\lim_{t\to\infty}A(t)=0$ such that
\begin{equation}\label{sced2ndord_cond}
\lim_{x\to x^*}
\sup_{n\in\N}\max_{1\leq i\leq n}\left|\frac{1-F_{i,s_j}(x)}{1-F_Z(x)}-c\left(\frac i n,s_j\right)\right|=
O\left[A\left\{\frac 1{1-F_Z(x)}\right\}\right]\quad\forall_{j=1,\ldots,m}.
\end{equation} 

Recall that a standard bivariate Wiener process $W(s,t)$, $(s,t)\in
(0,2]\times[0,1]$, is a Gaussian process with mean zero 
and 
 $\cov\{W(s_1,t_1)W(s_2,t_2)\}=(s_1\wedge s_2)\times(t_1\wedge t_2)$
 for $(s_1,t_1),(s_2,t_2)\in (0,2]\times[0,1]$.
In particular, $W(1,\cdot)$ and $W(\cdot, 1)$ are both standard univariate
Wiener processes.

\begin{theorem}\label{Casympt_teo}
	Suppose the second-order conditions \eqref{UZ2ndord_cond} and
        \eqref{sced2ndord_cond} hold. Recall that $c$ is a continuous positive function on $[0,1]\times S$ verifying $\sum_{j=1}^m \int_0^1 c(t,s_j)dt=1$. For $k=k(n)\to\infty$, $n/k\to 0$, and
	\[
	\surd k \left\{A\left(\frac n k\right)\vee \alpha\left(\frac N k\right)\vee \sup_{|u-v|\leq 1/n}\sup_{s\in S}\;\left|c(u,s)-c(v,s)\right|\right\}\to 0,
	\]
	as $n\to\infty$, it follows for all $j=1,\ldots,m$ that 
	\begin{equation}\label{limCjt}
	\surd k\left\{\hat C_j(t)-C_j(t)\right\}{\mathop {\longrightarrow}^{d} } W_j\left\{1,C_j(t)\right\}-C_j(t)\sum_{j=1}^m W_j\left\{1,C_j(t)\right\}.
	\end{equation}
Here, the $W_j$ are standard bivariate Wiener processes on $[0,1]^2$ .
\end{theorem}


\begin{proof}
	Similarly as in Einmahl, de Haan and Zhou (2016), 
	the sequential tail empirical process is constructed for each $j$-th sample (recall we have independence in time $i$) but analysed in the common tail region $b_0\{N/(ku)\}$, $0\leq u\leq 2$. Then, the proof of their Theorem 4 adapts to obtain in this case,
	\begin{equation}\label{stepexpanTheo4}
	\sup_{0<u\leq 2}\;\sup_{0\leq t\leq 1}u^{-\eta}\left| \surd k\left\{\frac 1 k \sum_{i=1}^{nt}1_{\{X_i(s_j)>b_0\left(\frac{N}{ku}\right)\}}-uC_j(t)\right\}-\tilde W_j\left\{u,C_j(t)\right\}\right|\to 0,
	\end{equation}
	almost surely (a.s.), $j=1,\ldots,m$, as $n\to\infty$, for any $0\leq\eta<1/2$, under a Skorokhod construction with $\tilde W_j$ on $[0,1]^2$ standard bivariate Wiener processes.
	
	Let $
        L(x) = b_0\left(\frac N k\right)+xa_0\left(\frac N k\right)
        $ and
	\begin{equation}\label{uN}
	u=u_N=\frac N k \left[ 1-F_Z\left\{L(x)\right\}\right].
	\end{equation}
	The second-order condition \eqref{UZ2ndord_cond} implies 
	\begin{equation}\label{uNOalpha}
	\frac{u_N}{(1+\gamma x)^{-1/\gamma}}=1+O\left\{\alpha\left(\frac N k\right)\right\},\qquad n\to\infty,
	\end{equation}
	uniformly for $x\geq x_0>-1/\gamma_+$, see
        Proposition \ref{1-FOA_prop}. Substituting \eqref{uN} and 
        \eqref{uNOalpha} in \eqref{stepexpanTheo4} yields
       	for $j=1,\ldots,m$  almost surely 
	\begin{multline}\label{stepexpan}
	\sup_{x\geq x_0>-1/\gamma_+}\;
        \sup_{0\leq t\leq 1}(1+\gamma x)^{\eta/\gamma}
	\left| \surd k\left\{\frac 1 k
            \sum_{i=1}^{nt}
            1_{\{X_i(s_j)>L(x)\}}
      -(1+\gamma x)^{-1/\gamma}C_j(t)\right\}\right.
\\
	 -\tilde W_j\left\{(1+\gamma
           x)^{-1/\gamma},C_j(t)\right\}\Bigg|\to 0
         \quad (n\to\infty).
       \end{multline}
 Recall that $\sum_{j=1}^n C_j(1)=1$. For $t=1$ and summing up \eqref{stepexpan} in $j$,
	\begin{multline}
	\sup_{x\geq x_0>-1/\gamma +}(1+\gamma x)^{\eta/\gamma}
	\Bigg| \surd k\left\{\frac 1 k
          \sum_{j=1}^{m}\sum_{i=1}^{n}
          1_{\{X_i(s_j)>L(x)\}}
          -(1+\gamma x)^{-1/\gamma}C_j(1)\right\}
        \\
	-\sum_{j=1}^{m}
        \tilde W_j\left\{(1+\gamma
          x)^{-1/\gamma},C_j(1)\right\}\Bigg|\to 0
        \label{FNlim}
	\end{multline}
	almost surely as $n\to\infty$. Applying Vervaat's lemma (Vervaat 1971),
	\[
	\left| \surd
          k\left\{\frac{X_{N-(kx),N}-
              b_0\left(\frac{N}{k}\right)}{a_0\left(\frac{N}{k}\right)}-
            \frac{x^{-\gamma}-1}\gamma\right\}-
          x^{-\gamma-1}\sum_{j=1}^{m}\tilde W_j\left\{x,C_j(1)\right\}\right|\to 0,
	\]
	almost surely as $n\to\infty$. Hence, for $x=1$, we have
        almost surely
	\begin{equation}\label{limXN-[kx],N}
	\left| \surd k
	\left\{\frac{X_{N-k,N}-b_0\left(\frac{N}{k}\right)}{a_0\left(\frac{N}{k}\right)}\right\}-\sum_{j=1}^{m}\tilde W_j\left(1,C_j(1)\right)
	\right|\to 0\quad (n\to\infty).
	\end{equation}
	
	On the other hand, substituting $k^{-1/2}\sum_{i=1}^{m}\tilde W_j\left\{1,C_j(1)\right\}(1\pm\delta)$ for $x$ in \eqref{stepexpan}, we get for $\delta>0$,
	\begin{multline*}
	\sup_{0\leq t\leq 1}\Bigg| \surd k
	\Bigg(
	  \frac 1 k \sum_{i=1}^{nt}1_{\{X_i(s_j)>b_0
	  \left(N/k\right)+\frac 1{\surd k}\sum_{i=1}^{m}\tilde W_j\left\{1,C_j(1)\right\}(1\pm\delta)a_0\left(N/k\right)\}}\\
	  -\left[
	    1-\frac 1{\surd k}\sum_{i=1}^{m}\tilde W_j\left\{1,C_j(1)\right\}
	   \right]
	   C_j(t)
	\Bigg)
	-\tilde W_j\left\{1,C_j(t)\right\}
	\Bigg|\to 0\quad (n\to\infty),
	\end{multline*}
	i.e.
	\begin{multline*}
	\sup_{0\leq t\leq 1}\Bigg| \surd k\left[\frac 1 k \sum_{i=1}^{nt}1_{\{X_i(s_j)>b_0\left(N/k\right)+\frac 1{\surd k}\sum_{i=1}^{m}\tilde W_j\left\{1,C_j(1)\right\}(1\pm\delta)a_0\left({N}/{k}\right)\}}-C_j(t)\right]\\
	-\tilde W_j\left\{1,C_j(t)\right\}+C_j(t)\sum_{i=1}^{m}\tilde W_j\left\{1,C_j(1)\right\}\Bigg|\to 0\quad (n\to\infty),
      \end{multline*}
      almost surely for $j=1,\ldots,m$. The result follows by combining this with \eqref{limXN-[kx],N}.
\end{proof}

\begin{remark}\label{location_rem}
	From \eqref{limXN-[kx],N} and $a_0(n/k)/b_0(n/k)\to\gamma_+$ (a more refined relation is \eqref{lim_a/U-g/A} below) it follows
	\[
	\frac{X_{N-k,N}}{b_0(\frac N k)}=1+O_p\left(\frac 1{\surd k}\right),\qquad n\to\infty.
	\]
\end{remark}

\begin{theorem}\label{g_asympt_teo}
	Suppose \eqref{UZ2ndord_cond} and \eqref{sced2ndord_cond} hold
        with $\gamma\neq\rho<0$ and $c$ from Theorem \ref{Casympt_teo}. Suppose $k\to\infty$, $n/k\to 0$ and 
	\begin{equation}\label{growthkn}
	\surd k \left[A\left(\frac n k\right)\vee \alpha\left(\frac N k\right)\vee \left\{\frac{a\left(\frac N k\right)}{U\left(\frac N k\right)}-\gamma_+\right\}\vee\sup_{|u-v|\leq 1/n}\sup_{s\in S}\;\left|c(u,s)-c(v,s)\right|\right]\to 0,
	\end{equation}
	as $n\to\infty$. Then,
	\begin{equation}\label{asympnormg}
	\surd k\left(\hat \gamma-\gamma\right){\mathop {\longrightarrow}^{d} } \left\{\gamma_+ +2\left(1-\gamma_-\right)^2\left(1-2\gamma_-\right)\right\}L_1-\frac{\left(1-\gamma_-\right)^2\left(1-2\gamma_-\right)^2}2 L_2,
	\end{equation}
	and
	\begin{multline}\label{asympnorma}
	\surd k\left\{\frac{\hat a_Z\left(\frac N k\right)}{a_Z\left(\frac N k\right)}-1\right\}{\mathop {\longrightarrow}^{d} }\\ \gamma_+\sum_{j=1}^{m} W_j\left\{1,C_j(1)\right\}+\left(1-\gamma_-\right)\left\{1-2
	\left(1-2\gamma_-\right)\right\}L_1
	+\frac{\left(1-\gamma_-\right)\left(1-2\gamma_-\right)^2}2 L_2,
	\end{multline}
	with $W_j$ on $[0,1]^2$ standard bivariate Wiener processes and,
	\begin{eqnarray*}
		L_1&=& \int_0^{1/\gamma_-} \sum_{j=1}^{m} W_j\left\{(1+\gamma x)^{-1/\gamma},C_j(1)\right\}\frac{dx}{1+\gamma_+ x}\\
		L_2&=& 2\int_0^{1/\gamma_-} (1+\gamma
                       x)^{-1/\gamma}\sum_{j=1}^{m}
                       W_j\left\{1,C_j(1)\right\} + {}\\
		&&\hspace{1cm}{}+\frac{\log\left(1+\gamma_+ x\right)}{\gamma_+}\sum_{j=1}^{m} W_j\left\{(1+\gamma x)^{-1/\gamma},C_j(1)\right\}\frac{dx}{1+\gamma_+ x}.
	\end{eqnarray*}
\end{theorem}

\begin{proof}
  Asymptotic normality will be proved with the auxiliary functions related to the uniform bounds of the second-order condition; recall \eqref{abA0}. 	Let
  \[
  F_N(x)=\frac 1 N \sum_{i=1}^{n} \sum_{j=1}^{m} 1_{\{X_i(s_j)\leq x\}},
  \]
  and
  $$
  L(x) = b_0\left(\frac N k\right)+xa_0\left(\frac N k\right)
  $$
  Then, for $l=1,2$,
  \begin{eqnarray*}
    \lefteqn{M_{N}^{(l)}=\frac{1}{k} \sum_{i=1}^{n}
    \sum_{j=1}^{m}\left\{\log X_{i}(s_j) - \log
    X_{N-k,N}\right\}^l 1_{\{X_i(s_j)>X_{N-k,N}\}}}\nonumber\\ 
 &=&\frac {N} k \int_{X_{N-k,N}}^\infty\left(\log x -
     \log X_{N-k,N}\right)^l dF_N(x)\nonumber\\ 
 &=&\frac {lN} k \int_{X_{N-k,N}}^\infty \left(\log x - \log
     X_{N-k,N}\right)^{l-1}\left\{1-F_N(x)\right\}
     \frac{dx}x\nonumber\\ 
 &=&\frac{l N} k
     \int_{\frac{X_{N-k,N}-b_0(N/k)}{a_0(N/k)}}^{1/\gamma_-}
     \left[\log\left\{ \frac{L(x)}{X_{N-k,N}} \right\}\right]^{l-1}
	\left[1-F_N\left\{L(x)\right\}\right] \frac{a_0\left(\frac N
     k\right) dx}{L(x)}\nonumber
  \end{eqnarray*}
using partial integration in the third equality and variable
substitution $x=L(y)$ in the last equality. We
split $M_{N}^{(l)}$ into a sum $I(l) + J(l)$
with 
 \begin{eqnarray*}\nonumber
   I(l) &=& \frac{l N} k \int_{\frac{X_{N-k,N}-b_0(N/k)}{a_0(N/k)}}^0
     \left[\log\left\{ \frac{L(x)}{X_{N-k,N}}
         \right\}\right]^{l-1}\left[1-F_N\left\{L(x)\right\}\right]\frac{a_0\left(\frac
         N k\right) dx}{L(x)}\nonumber
\\
   J(l) &= & \frac {lN} k \int_0^{1/\gamma_-}
                   \left[\log\left\{ \frac{L(x)}{X_{N-k,N}}
                   \right\}\right]^{l-1}
\left[1-F_N\left\{L(x)\right\}\right] \frac{a_0\left(\frac N k\right) dx}{L(x)}
  \nonumber\end{eqnarray*}
Now,
\begin{eqnarray*}
  \lefteqn{\sup_{0<x<\left|\frac{X_{N-k,N}-b_0(N/k)}{a_0(N/k)}\right|}\left|\log\left\{\frac{b_0\left(\frac N k\right)+xa_0\left(\frac N k\right)}{X_{N-k,N}}\right\}\right|}\\
&=&\sup_{0<x<\left|\frac{X_{N-k,N}-b_0(N/k)}{a_0(N/k)}\right|}
    \left|\log\left\{1+\frac{a_0\left(\frac N k\right)}{b_0\left(\frac N k\right)}x\right\}-\log\left\{1+\frac{a_0\left(\frac N k\right)}{b_0\left(\frac N k\right)}\cdot\frac{X_{N-k,N}-b_0\left(\frac N k\right)}{a_0\left(\frac N k\right)}\right\}\right|\\
&=&\frac{a_0\left(\frac N k\right)}{b_0\left(\frac N k\right)}O_P\left(\frac 1{\surd k}\right)
\end{eqnarray*}
Hence,
\begin{eqnarray*}
  I(2) = 	2\frac{a_0\left(\frac N k\right)}{
  b_0\left(\frac N k\right)}O_P\left(k^{-1/2}\right) I(1).
\end{eqnarray*}
We calculate $I(1)$ by splitting it into two summands $I_1$ and 
$I_2$,
\begin{eqnarray*}
  I_1 &=&\frac{a_0\left(\frac N k\right)}{b_0\left(\frac N k\right)}
  \int_{\frac{X_{N-k,N}-b_0(N/k)}{a_0(N/k)}}^0 \frac N k
  \left[1-F_N\left\{L(x)\right\} \right]-(1+\gamma
x)^{-1/\gamma}\frac{dx}{1+ x 
  a_0\left(\frac N k\right)/b_0\left(\frac N
    k\right)},
\\
I_2 
& = &\frac{a_0\left(\frac N k\right)}{b_0\left(\frac N k\right)}\int_{\frac{X_{N-k,N}-b_0(N/k)}{a_0(N/k)}}^0 (1+\gamma x)^{-1/\gamma}\frac{dx}{1+ x a_0\left(\frac N k\right)/b_0\left(\frac N k\right)}.
	\end{eqnarray*}
Combining \eqref{FNlim}--\eqref{limXN-[kx],N}, dominated convergence, the growth of  $k$ (being restricted mainly by \eqref{growthkn}), and
  the fact that \eqref{UZ2ndord_cond} with $\gamma\neq\rho$ implies
 (cf. de Haan and Ferreira 2006)
	\begin{equation}\label{lim_a/U-g/A}
	\lim_{t\to\infty}\frac{\frac{a_0(t)}{b_0(t)}-\gamma_+}{\alpha_0(t)}=
	\left\{\begin{array}{ll}
	0 \;,                         & \gamma<\rho\leq 0\\
	\pm\infty \;,                 & \rho<\gamma\le 0 
	\text{ or } (0<\gamma<-\rho \text{ and } l\neq 0)
	\text{ or } \gamma=-\rho\\
	\frac{\gamma}{\gamma+\rho}\;, & (0<\gamma<-\rho \text{ and } l=0) 
	\text{ or }\gamma>-\rho\geq 0
	\end{array}
	\right.
	\end{equation}
	with $l= \lim_{t\to\infty} U_Z(t)-a(t)/\gamma$, it follows that 
	\[
	\frac{b_0(t)}{a_0(t)}I_1=o_P\left(\frac 1{\surd k}\right).
	\]
	Moreover by \eqref{limXN-[kx],N} and again by the conditions on the growth of $k$, 
	\[
	\surd k\frac{b_0(t)}{a_0(t)}I_2{\mathop {\longrightarrow}^{d} }-\sum_{j=1}^{m}\tilde W_j\left\{1,C_j(1)\right\},
	\]
	hence
	\begin{equation}\label{Ilim}
	\surd k \frac{b_0(t)}{a_0(t)}I(1) {\mathop {\longrightarrow}^{d} } -\sum_{j=1}^{m}\tilde W_j\left(1,C_j(1)\right).
	\end{equation}
        By \eqref{FNlim}, \eqref{limXN-[kx],N}, dominated convergence theorem, the conditions on the growth of $k$ and \eqref{lim_a/U-g/A}. It follows,
	\begin{equation}\label{IVlim}
	\left\{\frac{b_0(t)}{a_0(t)}\right\}^2 I(2)=2\frac{b_0\left(\frac N k\right)}{a_0\left(\frac N k\right)}O_P\left(\frac 1{\surd k}\right) I(1)=o_P\left(\frac 1{\surd k}\right).
	\end{equation}

The second summand $J(l)$ can be treated as follows:
        \begin{eqnarray}
	\lefteqn{\surd k \left[\left\{\frac{b_0\left(\frac N
          k\right)}{a_0\left(\frac N k\right)}\right\}^lJ(l)-
          \prod_{k=1}^l \frac k{1-k\gamma_-}\right]}\nonumber\\
	&=&\surd k\left(\int_0^{1/\gamma_-} l\left[\frac{\log\left\{ \frac{L(x)}{X_{N-k,N}} \right\}}{a_0\left(\frac N k\right)/b_0\left(\frac N k\right)}\right]^{l-1}\frac N k \left[
	1-F_N\left\{L(x)\right\}\right]
            \frac{dx}{1+\frac{a_0\left(\frac N
            k\right)}{b_0\left(\frac N k\right)}x} - \prod_{k=1}^l \frac k{1-k\gamma_-}\right)\nonumber\\
	&=& J_1 - J_2 \nonumber
	\end{eqnarray}
with
      \begin{eqnarray}\nonumber
J_1 &=& 
l\surd k\Bigg(\int_0^{1/\gamma_-} \left[\frac{\log\left\{
            \frac{L(x)}{X_{N-k,N}} \right\}}{a_0\left(\frac N
            k\right)/b_0\left(\frac N k\right)}\right]^{l-1}\frac N k
        \left[1-F_N\left\{L(x)\right\}\right]
\\\nonumber
&&\hspace{6cm}{}-\frac{\left[\log\left\{1+(\gamma_+)
        x\right\}\right]^{l-1}}{\gamma_+^{l-1} (1+\gamma x)^{1/\gamma}}\Bigg)
 \frac{dx}{1+\frac{a_0\left(\frac N k\right)}{b_0\left(\frac N
        k\right)}x},
\\
J_2 &=&l\surd k \int_0^{1/\gamma_-}\frac{\left[\log\left\{1+(\gamma_+)
        x\right\}\right]^{l-1}}{\gamma_+^{l-1}(1+\gamma
        x)^{1/\gamma}}
\left[ \frac1{1+(\gamma_+) x } -\frac1{1+\frac{a_0\left(\frac N k\right)}{b_0\left(\frac N
        k\right)}x} \right]dx
	\end{eqnarray}
and $\log\left\{1+(\gamma_+) x\right\}/\gamma_+=x$ for $\gamma\leq 0$.

In case of $\gamma$ being positive, we get 
\[
J_2 = l\surd k \left\{ \gamma- \frac{a_0\left(\frac N
      k\right)}{b_0\left(\frac N k\right)}\right\}
\int_0^\infty 
\frac{\left[\log\left\{1+(\gamma_+)
        x\right\}\right]^{l-1}}{\gamma_+^{l-1}(1+\gamma
        x)^{1/\gamma+2}}
x\,\frac{1+\gamma x}{1+ x a_0\left(\frac N k\right)/b_0\left(\frac N k\right)} dx
	.\]
	It is easily seen that it converges to zero, 
        from the conditions on the growth of $k$ and \eqref{lim_a/U-g/A}, and since 
	\begin{equation}\label{uniflim}
	\frac{1+\gamma x}{1+ x a_0\left(\frac N k\right)/b_0\left(\frac N k\right)}\to 1\quad n\to\infty,
	\end{equation}
	uniformly for all $x\geq 0$. The case $\gamma \le 0$
        can be treated similarly.
        Let us turn to $J_1$. We have that 
	\begin{equation}
          \left[\log\left\{
              \frac{b_0\left(\frac N k\right)+xa\left(\frac N
                  k\right)}{X_{N-k,N}} \right\}\right] \Big /
          \left\{a_0\left(\frac N
              k\right)/b_0\left(\frac N
              k\right)\right\}
        -
        \left[\frac{\log\left\{1+(\gamma_+)
              x\right\}}{\gamma_+}\right]
	\end{equation}
        equals 
	\begin{multline*}
          \left[\log\left\{ 1+\frac{a_0\left(\frac N
                    k\right)}{b_0\left(\frac N k\right)}x
              \right\}\right]  \Big /
            \left\{a_0\left(\frac N k\right)/b_0\left(\frac N
                k\right)\right\}
          -\left[\frac{\log\left\{1+(\gamma_+)
                x\right\}}{\gamma_+}
          \right] + {}
\\
   {} -\frac{ 
 \log\left\{1+\frac{a_0\left(\frac N
                  k\right)}{b_0\left(\frac N
                  k\right)}\frac{X_{N-k,N}-b_0(N/k)}{a_0(N/k)}\right\}}{\left[a_0\left(\frac N k\right)/b_0\left(\frac
                N k\right)\right]}
	\end{multline*}
	where the difference of first two summands on the right hand side
        is of smaller 
        order,
        and the third term gives the main contribution.
        Convergence statement \eqref{limXN-[kx],N}  yields almost surely
	\begin{multline*}
	\sup_{0<x<1/\gamma_-}\Bigg|
	\frac{\surd k}{1+(\gamma_+) x} 
        \left(\left[{\log\left\{ \frac{L(x)}{X_{N-k,N}}
              \right\}}\right]\Big / 
          \left\{ \frac{a_0\left(\frac N k\right)}{b_0\left(\frac N
                k\right)} \right\} 
-\left[\frac{\log\left\{1+(\gamma_+)
     x\right\}}{\gamma_+}\right]\right) {}
\\
 	- \frac{1
}{\left\{1+(\gamma_+)
            x\right\}}\sum_{j=1}^{m}\tilde
        W_j\left\{1,C_j(1)\right\}\Bigg| \to 0\quad
        \;(n\to\infty). 
	\end{multline*}

	Hence we find that, almost surely as $n\to\infty$,
	\begin{multline}
	\sup_{0<x<1/\gamma_-}
	\Bigg|
	\frac{\surd k}{1+(\gamma_+) x} 
        \Bigg(\left[\log\left\{ \frac{L(x)}{X_{N-k,N}} \right\}\right]
          \Big/
\left\{ \frac{a_0\left(\frac N k\right)}{b_0\left(\frac N
      k\right)}\right\}
\frac N k \left[
	1-F_N\left\{L(x)\right\}\right]
      \\	
      - (1+\gamma x)^{-1/\gamma}
      \left[\frac{\log\left\{1+(\gamma_+) x\right\}}{\gamma_+}\right]\Bigg)
  	 -\frac{1
}{(1+\gamma x)^{1/\gamma}}\sum_{j=1}^{m}\tilde
         W_j\left\{1,C_j(1)\right\}
  \\
  -
         \frac{\log\left\{1+(\gamma_+) x\right\}}{\gamma_+}
         \sum_{j=1}^{m}\tilde W_j\left\{(1+\gamma
           x)^{-1/\gamma},C_j(1)\right\}
         \Bigg|\to 0. \label{Vauxlim}
	\end{multline}
	
	Therefore,
	\begin{multline}\label{IIIlim}
	\surd k \left\{\frac{b_0\left(\frac N k\right)}{a_0\left(\frac N k\right)}M_n^{(1)}-\frac 1{1-\gamma_-}\right\}
	\\\nonumber
	{\mathop {\longrightarrow}^{d} }\int_0^{1/\{(-\gamma)\vee 0\}} \sum_{j=1}^{m}\tilde W_j\left\{(1+\gamma x)^{-1/\gamma},C_j(1)\right\}\frac{dx}{1+(\gamma_+) x}-\sum_{j=1}^{m}\tilde W_j\left\{1,C_j(1)\right\}.
	\end{multline}

        Furthermore, both
	\begin{eqnarray*}
	\surd k \left[\left\{\frac{b_0\left(\frac N
          k\right)}{a_0\left(\frac N k\right)}\right\}^2 J(2)-
          \frac 2{(1-\gamma_-)(1-2\gamma_-)}\right]
        \end{eqnarray*}
        and
	\begin{eqnarray*}
	\surd k\left[\left\{ \frac{b_0\left(\frac N
                k\right)}{a_0\left(\frac N k\right)} \right\}^2
          M_{n}^{(2)}- \frac 2{(1-\gamma_-)(1-2\gamma_-)}\right] 
       \end{eqnarray*}
       converge in distribution to
	\begin{multline}
        2\int_0^{1/\gamma_-} 
(1+\gamma x)^{-1/\gamma}
\sum_{j=1}^{m}\tilde
         W_j\left\{1,C_j(1)\right\}  +{}
\\\nonumber {} +
         \frac{\log\left\{1+(\gamma_+) x\right\}}{\gamma_+}
         \sum_{j=1}^{m}\tilde W_j\left\{(1+\gamma
           x)^{-1/\gamma},C_j(1)\right\}
        \frac{dx}{1+(\gamma_+) x}
	\end{multline}
	as $n\rightarrow\infty$. 
        By straightforward calculations applying Cram\'er's delta method result \eqref{asympnormg} follows.
	
	Result \eqref{asympnorma} now follows in a straightforward way from decomposition
	\[
	\frac{\hat a_Z\left(\frac N k\right)}{a_0\left(\frac N k\right)}=\frac{X_{N-k,N}}{b_0\left(\frac N k\right)}\frac 1 2M_N^{(1)}\frac{b_0\left(\frac N k\right)}{a_0\left(\frac N k\right)}\left\{1-\frac{\left(M_N^{(1)}\right)^2}{M_N^{(1)}}\right\}^{-1}
	\]
	and applying the previous limiting relations. 
	
	Finally, 
from \eqref{abA0} the same distributional results hold with $a_Z$ instead of $a_0$.
\end{proof}

\section*{Acknowledgement}
Research is partially funded by the VolkswagenStiftung support for Europe within the WEX-MOP project (Germany), and FCT - Portugal, projects UID/MAT/00006/2013 and UID/Multi/04621/2013.

\end{document}